\documentclass[11pt,draftcls,onecolumn]{IEEEtran}

\usepackage{color}
\usepackage[T1]{fontenc}		
\usepackage[utf8]{inputenc}	

\usepackage{graphicx}
\usepackage{amsmath,amssymb}
\usepackage{amsfonts}
\usepackage{amsthm}

\usepackage{subcaption}
\usepackage{nicefrac}
\usepackage{placeins} 
\usepackage[ruled,vlined,linesnumbered]{algorithm2e}
\usepackage[dvipsnames]{xcolor}

\newtheorem{remark}{Remark}
\newtheorem{theorem}{Theorem}
\newtheorem{lemma}{Lemma}
\newtheorem{prop}{Proposition}
\newtheorem{coro}{Corollary}

\usepackage{hyperref} 
\hypersetup{
colorlinks=true,
linkcolor=red,
citecolor=green
}

\usepackage{url}

	\newcommand{\dd}{\mathrm{d}}

\normalsize

\ifCLASSINFOpdf
\else
\fi

\begin{document}
%
\title{On the Transmission Probabilities in Quantum Key Distribution Systems over FSO Links}
%
%
%

\author{Hui Zhao, and Mohamed-Slim Alouini \vspace{-1cm}
\thanks{This paper was accepted for publication in IEEE Transactions on Communications on Oct. 4, 2020.}
\thanks{This work was funded by the office of sponsored research (OSR) at KAUST, and  the European Research Council under the EU Horizon 2020 research and innovation program/ERC grant agreement no. 725929 (ERC project DUALITY).}

\thanks{H. Zhao was with the Computer, Electrical, and Mathematical Science and Engineering Division, King Abdullah University of Science and Technology, Thuwal 23955-6900, Saudi Arabia, and he is now with the Communication Systems Department, EURECOM, Sophia Antipolis 06410, France (email: hui.zhao@kaust.edu.sa).}
\thanks{M.-S. Alouini is with the Computer, Electrical, and Mathematical Science and Engineering Division, King Abdullah University of Science and Technology, Thuwal 23955-6900, Saudi Arabia (email: slim.alouini@kaust.edu.sa).}
\thanks{Digital Object Identifier 10.1109/TCOMM.2020.3030250}
}

\maketitle

\begin{abstract}
In this paper, we investigate the transmission probabilities in three cases (depending only on the legitimate receiver, depending only the eavesdropper, and depending on both legitimate receiver and eavesdropper) in quantum key distribution (QKD) systems over free-space optical links. To be more realistic, we consider a generalized pointing error scenario, where the azimuth and elevation pointing error angles caused by stochastic jitters and vibrations in the legitimate receiver platform are independently distributed according to a non-identical normal distribution.  Taking these assumptions into account, we derive approximate expressions of transmission probabilities by using the Gaussian quadrature method. To simplify the expressions and get some physical insights, some asymptotic analysis on the transmission probabilities is presented based on  asymptotic expression for the generalized Marcum Q-function when the telescope gain at the legitimate receiver approaches to infinity. Moreover, from the asymptotic expression for the generalized Marcum Q-function, the asymptotic outage probability over Beckmann fading channels (a general channel model including Rayleigh, Rice, and Hoyt fading channels) can be also easily derived when the average signal-to-noise ratio is sufficiently large, which shows the diversity order and array gain.
\end{abstract}

\begin{IEEEkeywords}
Beckmann distribution, free-space optics, generalized pointing errors, quantum key distribution, and transmission probability.
\end{IEEEkeywords}

%
\IEEEpeerreviewmaketitle

\section{Introduction}
%
%
%
%
\IEEEPARstart{Q}{uantum} communication provides a promising solution to break the Shannon channel capacity limit \cite{Shannon} and achieve an unprecedented level of security \cite{Lo_Science} simultaneously, two competing tasks which cannot be realized in conventional technologies \cite{Sasaki}. In this context, quantum key distribution (QKD) or quantum cryptography is a method for sharing the secret cryptographic keys between two legitimate parties to achieve the secure communications by taking advantage of the laws of quantum mechanics and quantum non-cloning theorem \cite{Lo,Tittel}.  However, further investigation and real application of QKD did not attract much attention until it was proved that the quantum computer was able  to break public-key cryptosystems, which are commonly used in  the modern cryptography \cite{Inoue,Cheng}.

The connection implementation of QKD includes two main medium, i.e., fiber cable and free-space optics (FSO). Compared to the fiber cable, the implementation of QKD over FSO links is more convenient and easier due to the flexibility of the free space connection and satellite support for distributing quantum keys worldwide \cite{Selman,Trinh}. Moreover, FSO is an alternative transport technology to interconnect high capacity networking segments in current and future communication systems, because of its cost-effectiveness, high-bandwidth availability, and interference-immunity \cite{Imran}. The authors in \cite{R1}--\cite{R8} have presented some basic performance analysis works over classical FSO links or hybrid RF-FSO links, but those previous works do not consider the QKD mechanism. Actually, the research work about QKD over FSO links in the communications field is considerably limited.

In practical systems, the security of QKD strongly depends on the device implementation \cite{Xu}--\cite{Meda}. That is, a third party may have a side channel by  making use of any deviation of a QKD device from the theoretical model. For example, two zero-error attacks on commercial QKD systems were reported where the defects in quantum signal encoding and detection were exploited \cite{Xu,Lydersen}. Besides, some imperfections in QKD designs can be also exploited by a plethora of quantum hacking attacks using current technologies \cite{Weier}.

In real commercial QKD implementations, a single-photon mechanism is typically used to convey the information, and the corresponding common detection scheme is called single photon avalanche photodiodes (SPADs) where the SPAD diode is operated in Geiger mode (reverse-biased above the breakdown voltage to create an avalanche) to count single-photons \cite{Shi}. However, this detection scheme can lead to the information leakage, because the avalanche created by the incoming photon can emit a secondary photon which may be intercepted by a third party, namely the eavesdropper. This secondary photon emission (the photon emission in the sender is the first emission) is called backflash, which is quenched along with the avalanche, i.e., the backflash is quenched if the detection bias is lowered below the breakdown voltage \cite{Meda,Shi}. Previous measurements show that the probability of detecting backflash is greater than $0.4 \%$, and more than $0.04$ photons emerging from the devices are contained in the backflash given the $10 \%$ nominal detection efficiency of SPADs \cite{Meda,Lacaita}. These measurements provide a reference for the backflash resulting in the information leakage, although the measurement results may change in different detector types and optical components.

An unevadable vulnerability in the FSO QKD systems is the random pointing error due to stochastic jitters and vibrations which can be caused by building sway, thermal expansion, and week earthquakes, in the urban FSO systems \cite{Arnon}--\cite{Yang}. Similarly, for satellite communications, there are internal and external reasons for stochastic vibrations \cite{Arnon_proc}. For example, the structure deformations caused by temperature gradients, and the gravitational force inhomogeneity over the satellite orbit, are two main external reasons for  stochastic vibrations in  satellite systems. The internal sources include electronic noise, antenna pointing operation, and solar array driver \cite{Abdi}.

The authors in \cite{Kupferman} first investigated the performance of received powers at both the legitimate receiver and eavesdropper in FSO QKD systems with taking random pointing errors into account, and derived the closed-form expressions for the corresponding average received powers. However, the authors in  \cite{Kupferman} assumed that the azimuth and elevation pointing errors are identically independently distributed, and more specifically, these two pointing errors are modeled by Gaussian distribution with zero-mean and the same variance which may be a little ideal. In the practical systems, the mean and variance of these two pointing errors are typically different. Moreover, the authors in \cite{Kupferman} did not consider the transmission probability depending on the received power threshold, which is very important and useful for the system evaluation and design. This is because we need to know the transmission probability depending on some conditions in the average level, apart from the average received powers, when evaluating and designing FSO QKD systems.

Actually, the pointing error angle, divided into the azimuth and elevation pointing errors, can be modeled by the Beckmann distribution \cite{Juan_TCOM}--\cite{Slim_book}, a generalized model including the Rayleigh, Hoyt and Rice models. Specifically, the Beckmann model is reduced into the Hoyt model for zero-mean and different variance of two sub-part pointing errors, and the different non-zero mean and equal variance case refers to the Rice case in the Beckmann distribution. The zero-mean and equal variance case considered in \cite{Kupferman} is the most simple scenario in the Beckmann distribution, denoted by the Rayleigh case. The authors in \cite{Zhao_CL_QKD} expanded the pointing error model in \cite{Kupferman} to the Beckmann distribution. In \cite{Zhao_CL_QKD}, exact closed-form expressions for the average received powers  at both the legitimate receiver and eavesdropper were derived, as well as finding the maximum points of the telescope gain at the legitimate receiver in some special cases analytically. However, the authors in \cite{Zhao_CL_QKD} still did not investigate the transmission probabilities depending on a variety practical conditions, which is also important for the FSO QKD  system evaluation.  Similar to \cite{Nguyen}--\cite{Jiang}, the transmission probability in this paper is defined as the probability that the received power  is satisfied one or more pre-set thresholds in the FSO QKD system, which is  obviously a natural variant of the outage probability in traditional communications \cite{Slim_book}.

Motivated by observing those facts outlined above, we investigate the performance of QKD systems over FSO links in terms of transmission probabilities depending on three different conditions. The main contributions of this paper are summarized as follows:
\begin{enumerate}
\item Closed-form expressions with a high accuracy for transmission probabilities depending on three different conditions, i.e., legitimate receiver, eavesdropper and both legitimate receiver and eavesdropper, are derived based on the Gaussian quadrature rule, where the accuracy grows with increasing the summation terms in the Gaussian quadrature.

\item Asymptotic expression for the generalized Marcum Q-function is derived after some mathematical manipulations, which can be used to derive the asymptotic outage probability over Beckmann fading channels when the average signal-to-noise ratio (SNR) approaches to infinity, showing the diversity order and array gain, since those two metrics govern the outage probability behaviour in high SNRs.

\item By using the asymptotic result for the generalized Marcum Q-function, the asymptotic expressions for three transmission probabilities are easily derived, which are valid in the high value region of the telescope gain at the legitimate receiver. Besides providing some insights, these asymptotic expressions are significantly concise, resulting in a much faster calculation than the analytical expressions that need to be computed based on the Gaussian quadrature rule.

\item We also present some specific expressions for those three transmission probabilities in some simplified cases which result in exact expressions or more concise forms. More specifically, exact closed-form expressions in Rayleigh, Hoyt and Rice cases (three special cases of the Beckmann distribution) for the transmission probability depending only on the legitimate receiver are given. In the Rayleigh case, we present a more concise expression for the transmission probability depending  on both the legitimate receiver and eavesdropper.
\end{enumerate}

The remainder of this paper is organized as follows. The system model is presented in Section \ref{model}. The transmission probabilities depending on three different conditions are analyzed in Sections \ref{TPLR_sec}, \ref{TPE_sec} and \ref{TPRE_sec}, respectively. In Section \ref{numerical_sec}, some numerical results are generated and used to validate the correctness of derived closed-form expressions, as well as presenting some interesting comparisons. Section \ref{conclude_sec} finally concludes the paper.

\section{System Model}\label{model}
\begin{figure}
\setlength{\abovecaptionskip}{0pt}
\setlength{\belowcaptionskip}{10pt}
\centering
\includegraphics[width= 2.5 in]{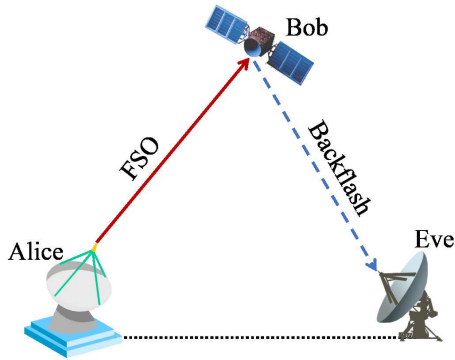}
\caption{Secure QKD System Over FSO links}
\label{system}
\end{figure}

As depicted in Fig. \ref{system}, there is a sender (Alice) located on a absolutely static platform\footnote{We can also consider a non-static platform for the sender. In a relative motion aspect, if the sender is assumed to be relatively static to the receiver, this will induce the same analysis.} communicating a legitimate receiver (Bob) located on a platform suffering from stochastic jitters and vibrations, such as a laser satellite system, over a FSO link in open areas. This vibrating platform in the legitimate receiver results in a random pointing error, where the stochastic deviation angle ($\theta$) is divided into two parts, i.e., the azimuth pointing error ($\theta_H$) and the elevation pointing error ($\theta_V$), and therefore, $\theta$ can be written as
\begin{align}
\theta  = \sqrt {\theta _H^2 + \theta _V^2}, \notag
\end{align}
where $\theta_H$ and $\theta_V$ are normally assumed to be independent Gaussian random variables, i.e., $\theta_H \sim \mathcal{N}(\mu_H,\sigma_H^2)$ and $\theta_V \sim \mathcal{N}(\mu_V, \sigma_V^2)$, where $\mu_H$ and $\sigma_H^2$ (or $\mu_V$ and $\sigma_V^2$) represent the mean and variance of $\theta_H$ (or $\theta_V$), respectively.

The SPADs detection scheme adopted by the legitimate receiver is assumed in this system setting. The  received power at Bob is\footnote{This power can be also regarded as an average received power over instantaneous received photon counts influenced by both the shot noise and the dead time of the SPAD receiver. Here, we focus on the transmission probability depending on the average power, rather than the instantaneous performance.} \cite{Arnon,Arnon_proc,Kupferman} 
\begin{align}
{P_D} (\theta) = {K_1}L\left( \theta  \right){G_D}, \notag
\end{align}
where  $L(\theta)=\exp(-G_D \theta^2)$ is the pointing loss factor, $G_D=(\pi d_D/ \lambda_1)^2$ is the telescope gain at the legitimate receiver, $d_D$ is the unobscured circular aperture diameter of the telescope, $\lambda_1$ is the wavelength, and $K_1$ is a constant depending only on the system design, given by
\begin{align}
{K_1} = {\eta _q}{P_S}{G_S}{\eta _S}{\eta _D}\frac{{{L_A}\left( {{Z_1}} \right)}}{{Z_1^2}}{\left( {\frac{\lambda_1 }{{4\pi }}} \right)^2}, \notag
\end{align}
in which $P_S$ is the optical transmitter power, $\eta_q$ is the quantum efficiency, $G_S$ is the telescope gain of the sender, $L_A(\cdot)$ is the atmospheric loss, and $Z_1$ is the distance between the sender and legitimate receiver.

As discussed in the introduction section, the SPADs detection scheme leads to the backflash due to the secondary photon emission caused by the avalanche. A third party (eavesdropper, Eve) can make use of this backflash to intercept the secondary photon, and thereby wiretapping the conveyed information from the sender to the legitimate receiver. The received power at the eavesdropper is given by \cite{Kupferman,Zhao_CL_QKD}
\begin{align}
{P_E}\left( {{\theta _E},\alpha } \right) = {K_2}{P_D}\left( \theta  \right)L\left( {{\theta _E}} \right){G_D}, \notag
\end{align}
where ${\theta _E} = \sqrt {{{\left( {{\theta _V} + \alpha } \right)}^2} + \theta _H^2}$, $\alpha$ is the pointing direction error angle in the wiretap FSO link, and $K_2$ is a system constant, given by
\begin{align}
{K_2} = {\eta _B}{\eta _q}{\eta _D}{\eta _E}{G_E}\frac{{{L_A}\left( {{Z_2}} \right)}}{{Z_2^2}}{\left( {\frac{{{\lambda_2 }}}{{4\pi }}} \right)^2}, \notag
\end{align}
in which $\eta_B$ is the probability of backflash, $Z_2$ is the distance between the legitimate receiver and eavesdropper, $\eta_E$ and $G_E$ are the optical efficiency and telescope gain of the eavesdropper respectively, and $\lambda_2$ is the backflash wavelength.

\section{Transmission Probability Depending only on Legitimate Receiver}\label{TPLR_sec}
In this section, we want to evaluate the transmission probability performance given a received power threshold ($\lambda_D$) at the legitimate receiver.
In this context, the transmission probability  depending only on the legitimate receiver (TPLR) is defined as
\begin{align}
{\rm{TPLR}} = \Pr \Big\{ {{P_D}\left( \theta  \right) \ge {\lambda _D}} \Big\} = \Pr \Big\{ {{K_1}{G_D}L\left( \theta  \right) \ge {\lambda _D}} \Big\},
\end{align}
which can be further written by substituting the expression for $L(\theta)$, given by
\begin{align}\label{TPLR_def}
{\rm{TPLR}} = \Pr \left\{ {{\theta ^2} \le \frac{{ - 1}}{{{G_D}}}\ln \frac{{{\lambda _D}}}{{{K_1}{G_D}}}} \right\}.
\end{align}
Let ${\Theta _D} = \frac{{ - 1}}{{{G_D}}}\ln \frac{{{\lambda _D}}}{{{K_1}{G_D}}}$. By substituting the probability density functions (PDFs) of $\theta_V$ and $\theta_H$ into \eqref{TPLR_def}, the TPLR for arbitrary  $\mu_V$, $\mu_H$, $\sigma_V$ and $\sigma_H$ can be derived as
\begin{align}
&{\rm{TPLR}} = \Pr \Big\{ {{\theta ^2} \le {\Theta _D}} \Big\} = \Pr \left\{ {\sqrt {\theta _V^2 + \theta _H^2}  \le \sqrt {{\Theta _D}} } \right\}  \notag\\
&= \iint\limits_{{\sqrt {\theta _V^2 + \theta _H^2}  \le \sqrt {{\Theta _D}} }} f_{\theta_V} (\theta_V) f_{\theta_H} (\theta_H) \dd \theta_V \dd \theta_H 
 = \iint\limits_{{\sqrt {\theta _V^2 + \theta _H^2}  \le \sqrt {{\Theta _D}} }} {\frac{\exp \left( { - \frac{{{{\left( {{\theta _V} - {\mu _V}} \right)}^2}}}{{2\sigma _V^2}} - \frac{{{{\left( {{\theta _H} - {\mu _H}} \right)}^2}}}{{2\sigma _H^2}}} \right)}{{2\pi {\sigma _V}{\sigma _H}}}}    \dd{\theta _V}\dd{\theta _H},
\end{align}
where the double integral, unfortunately, cannot be solved in a closed-form, and therefore, there are some approximation methods for this double integral, such as \cite{Juan_TCOM} and \cite{Zhu}. However, those approximation methods proposed by \cite{Juan_TCOM} and \cite{Zhu} are still complicated for calculation.

Here, we provide another approximation method based on the Gaussian quadrature rule \cite[Ch. 9]{Gaussian_Quadrature}, shown in Theorem \ref{TPLR_theorem}.

\begin{theorem}\label{TPLR_theorem}
An approximate result for the TPLR based on the Gaussian quadrature rule  is
\begin{align}\label{TPLR_final}
{\rm TPLR} \approx 1 - \frac{1}{{\sqrt \pi  }}\sum\nolimits_{i = 1}^N {{\omega _i}{f_{\rm TPLR}}\left( {{x_i}} \right)},
\end{align}
where $N$, $\omega_i$, and $x_i$ are the summation terms, weights, and selected points of the Gauss-Hermite quadrature (GHQ, a special case of Gaussian quadrature), and $f_{\rm TPLR} (\cdot)$ is given by
\begin{align}
&{f_{{\rm{TPLR}}}}\left( x \right) = {Q_{\frac{1}{2}}}\Biggm( \sqrt \lambda  , \frac{1}{\sigma _H} \sqrt {{{\Theta _D} - {{\left( {\sqrt 2 {\sigma _V}x + {\mu _V}} \right)}^2}}}  
  \times  \mathbb{I}\left\{ {{\Theta _D} \ge {{\left( {\sqrt 2 {\sigma _V}x + {\mu _V}} \right)}^2}} \right\} \Biggm),
\end{align}
in which $\lambda=\mu_H^2 / \sigma_H^2$ represents the  noncentrality parameter, $Q_{\cdot}(\cdot,\cdot)$ denotes the generalized Marcum Q-function \cite{Gradshteyn}. and $\mathbb{I}\{\cdot\}$ denotes the indicator function, i.e.,
$
\mathbb{I}\left\{ \mathcal{A} \right\} = \begin{cases}
{1,}&{\text{if } \mathcal{A} \text{ is true};}\\
{0,}&{\text{otherwise}.}
\end{cases}
$
\end{theorem}
\begin{proof}
See Appendix \ref{proof_TPLR_theorem}.
\end{proof}

\begin{remark}
Although a high accuracy requires many terms in \eqref{TPLR_final}, resulting in a much slower calculation, especially when $Q_{\frac{1}{2}}(\cdot)$ cannot be directly calculated in some softwares, such as Matlab, Theorem \ref{TPLR_theorem} provides an analytical tool to investigate the TPLR. We will present an asymptotic expression presented in the III-A subsection, rather than \eqref{TPLR_final}, for getting a high accuracy result in a special case.
\end{remark}

As $\theta_V$ follows the Gaussian distribution, for $\sigma_V^2 \ll \sigma_H^2$, according to \cite[Eq. (4)]{Pan_TVT}, the TPLR in \eqref{TPLR_integral} can be robustly approximated by
\begin{align}
{\rm{TPLR}} &= {\mathbb{E}_{{\theta _V}}}\left\{ {1 - {Q_{\frac{1}{2}}}\left( {\sqrt \lambda  ,\sqrt {\frac{{{\Theta _D} - \theta _V^2}}{{\sigma _H^2}}\mathbb{I}\left\{ {{\Theta _D} \ge \theta _V^2} \right\}} } \right)} \right\} \notag\\
&\mathop  \approx \limits^{\sigma _V^2 \ll \sigma _H^2} \frac{2}{3}{\phi _{{\rm{TPLR}}}}\left( {{\mu _V}} \right) + \frac{1}{6}{\phi _{{\rm{TPLR}}}}\left( {{\mu _V} + \sqrt 3 {\sigma _V}} \right) 
+ \frac{1}{6}{\phi _{{\rm{TPLR}}}}\left( {{\mu _V} - \sqrt 3 {\sigma _V}} \right),
\end{align}
where
\begin{align}
{\phi _{{\rm{TPLR}}}}\left( x \right) = 1 - {Q_{\frac{1}{2}}}\left( {\sqrt \lambda  ,\sqrt {\frac{{{\Theta _D} - {x^2}}}{{\sigma _H^2}}\mathbb{I}\left\{ {{\Theta _D} \ge {x^2}} \right\}} } \right).
\end{align}
This robust approximation for $\sigma_V^2 \ll \sigma_H^2$ was proposed by \cite{Pan_TVT}--\cite{Zhao_CL}.  This robust result becomes more tighter along with the ratio of $\sigma_V^2 / \sigma_H^2$ approaching to zero.

\subsection{Asymptotic Result for TPLR as $\Theta_D \to 0$}
Before presenting the asymptotic analysis for TPLR, we first give the following proposition.
\begin{prop}\label{Q_asym_prop}
The asymptotic expression for $Q_M(a,b)$, the generalized Marcum Q-function, as $b \to 0$, is given by
\begin{align}\label{Prop_31}
{Q_M}\left( {a,b} \right)\mathop  \simeq \limits^{b \to 0} 1 - \frac{{\exp \left( { - \frac{{{a^2}}}{2}} \right)}}{{\Gamma \left( {M + 1} \right){2^M}}}{b^{2M}}+o\left( {{b^{2M + 1}}} \right),
\end{align}
where $M$, $a$, $b$ are non-negative, $o(\cdot)$ and $\Gamma(\cdot)$ denote the higher order term and Gamma function \cite{Gradshteyn}, respectively.
\end{prop}

\begin{proof}
See Appendix \ref{proof_marcumQ_prop}.
\end{proof}

We present some numerical results in Fig. \ref{Marcum} to validate the correctness of the derived asymptotic expression for the generalized Marcum Q-function. It is obvious that the asymptotic results match the exact results very well when $b$ is sufficiently small.
\begin{figure}[!htb]
\vspace{-0.5cm}
\setlength{\abovecaptionskip}{0pt}
\setlength{\belowcaptionskip}{10pt}
\centering
\includegraphics[width=3 in]{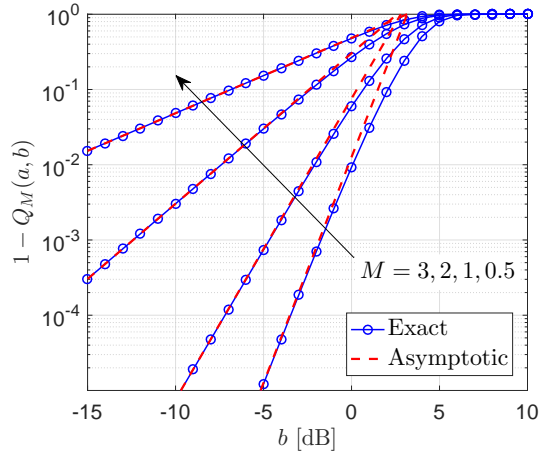}
\caption{$1-Q_M(a,b)$ versus $b$ for $a=1$.}\vspace{-0.5cm}
\label{Marcum}
\end{figure}

For $G_D \to +\infty$, we have
\begin{align}
{\Theta _D} = \frac{{ - 1}}{{{G_D}}}\ln \frac{{{\lambda _D}}}{{{K_1}{G_D}}} \to {0^ + }, \quad \text{as }{G_D} \to  + \infty.
\end{align}
In the following, we will analyze the asymptotic behaviour of TPLR when $\Theta_D \to 0^+$ (or equivalently $G_D \to +\infty$).

\begin{lemma}\label{TPLR_asym_lemma}
For $\Theta_D \to 0$ (or equivalently $G_D \to +\infty$), by using  Proposition \ref{Q_asym_prop}, the asymptotic expression for TPLR can be derived as
\begin{align}\label{TPLR_asym}
{\rm{TPLR}}\mathop  \simeq \limits^{{\Theta _D} \to 0} \frac{{\sqrt \pi  }}{{4\Gamma \left( {1.5} \right){\sigma _V}{\sigma _H}}}\exp \left( { - \frac{\lambda }{2} - \frac{{{\mu _V^2}}}{{2\sigma _V^2}}} \right){\Theta _D}.
\end{align}
\end{lemma}

\begin{proof}
See Appendix \ref{proof_TPLR_asym_lemma}.
\end{proof}

\begin{remark}
In the $G_D$ saturation case, a linear mapping from $\Theta_D$ to the TPLR  is derived, shown in Lemma \ref{TPLR_asym_lemma}, which is interesting in the performance analysis aspect. This asymptotic expression can not only simplify the TPLR calculation significantly, but also reveal the relationship between the TPLR and $\Theta_D$ in the $G_D$ saturation case.
\end{remark}

\begin{remark}
In fact, the asymptotic expression presented in Lemma \ref{TPLR_asym_lemma} can be also viewed as the asymptotic result for the outage probability over Beckmann fading channels (including Rayleigh, Hoyt and Rice fading channels), where $\Theta_D$ represents the received SNR threshold, and the instantaneous SNR at the receiver is $\gamma= \theta_V^2+\theta_H^2$.
\end{remark}

Although the atmospheric turbulence is not a main investigation in this paper (we will consider this issue in our future work),  we can simply analyze this impact  on the TPLR in the $G_D$ saturation case based on Lemma \ref{TPLR_asym_lemma}, shown in Corollary \ref{turbulence_coro} where the turbulence is modeled by a Gamma-Gamma distribution, a widely adopted turbulence model   \cite{Hessa_TWC}.

\begin{coro}\label{turbulence_coro}
In the $G_D$ saturation case, if the atmospheric turbulence modeled by a Gamma-Gamma distribution is considered over the FSO link,  the asymptotic result for the TPLR is 
\begin{align}\label{TPLR_Turbulence}
{\rm TPLR}=&\frac{{\sqrt \pi  }}{{4  \Gamma \left( {1.5} \right){\sigma _V}{\sigma _H}}}\exp \left( { - \frac{\lambda }{2} - \frac{{{\mu _V^2}}}{{2\sigma _V^2}}} \right) \Theta_D  \notag\\
&\hspace{4cm}+ \underbrace{\frac{{\sqrt \pi \left( \psi(\alpha_D)+\psi(\beta_D) -\ln (\alpha_D\beta_D)\right) }}{{4G_D \Gamma \left( {1.5} \right){\sigma _V}{\sigma _H}}}\exp \left( { - \frac{\lambda }{2} - \frac{{{\mu _V^2}}}{{2\sigma _V^2}}} \right)}_{\text{Atmospheric Turbulence}},    
\end{align}
where $\psi(\cdot)$ denotes the digamma function \cite{Gradshteyn},  $\alpha_D$ and $\beta_D$ are the fading parameters
of large-scale and small-scale fluctuations, respectively.
\end{coro}

\begin{proof}
It is obvious that the TPLR in \eqref{TPLR_def} becomes
\begin{align}
{\rm TPLR} &=\Pr \Big\{I_D K_1 G_D L(\theta) \ge \lambda_D\Big\}=
\Pr\left\{\theta^2 \le \frac{-1}{G_D} \ln \frac{\lambda_D}{I_D K_1 G_D}\right\} \notag\\
&=\Pr\left\{\theta^2 \le \Theta_D +\frac{1}{G_D} \ln I_D\right\},    
\end{align}
where $I_D$ represents the atmospheric turbulence following a Gamma-Gamma distribution with the PDF \cite[Eq. (11)]{Hessa_TWC},
\begin{align}
f_{I_D}(x)=\frac{2(\alpha_D \beta_D)^{\frac{\alpha_D+\beta_D}{2}}}{\Gamma(\alpha_D)\Gamma(\beta_D)}  x^{\frac{\alpha_D+\beta_D}{2}-1} K_{\alpha_D-\beta_D} \left(2 \sqrt{\alpha_D \beta_D x}\right),  
\end{align}
where $K_{\cdot}(\cdot)$ denotes the modified Bessel function of the second kind \cite{Gradshteyn}.

Define $\Theta_D^\prime = \Theta_D +\frac{1}{G_D} \ln I_D$. In view of Lemma \ref{TPLR_asym_lemma}, it is easy to derive the asymptotic result for the TPLR as
\begin{align}\label{TPLR_turbulence_int}
&{\rm{TPLR}} \mathop  \simeq \limits^{{\Theta _D} \to 0} \frac{{\sqrt \pi  }}{{4\Gamma \left( {1.5} \right){\sigma _V}{\sigma _H}}}\exp \left( { - \frac{\lambda }{2} - \frac{{{\mu _V^2}}}{{2\sigma _V^2}}} \right) \cdot\mathbb{E}_{I_D}\left\{\Theta _D^\prime\right\} \notag\\
&= \frac{{\sqrt \pi  }}{{4  \Gamma \left( {1.5} \right){\sigma _V}{\sigma _H}}}\exp \left( { - \frac{\lambda }{2} - \frac{{{\mu _V^2}}}{{2\sigma _V^2}}} \right) \Theta_D + \underbrace{\frac{{\sqrt \pi  }}{{4G_D \Gamma \left( {1.5} \right){\sigma _V}{\sigma _H}}}\exp \left( { - \frac{\lambda }{2} - \frac{{{\mu _V^2}}}{{2\sigma _V^2}}} \right) \cdot\mathbb{E}_{I_D}\left\{ \ln I_D\right\}}_{\text{Atmospheric Turbulence}}.
\end{align}
By using the PDF of $I_D$, the expectation of $\ln I_D$ with respect to $I_D$ can be obtained as \cite{Hessa_TWC}
\begin{align}\label{Ln_ID}
\mathbb{E} \left\{\ln I_D\right\} &=  \frac{2(\alpha_D \beta_D)^{\frac{\alpha_D+\beta_D}{2}}}{\Gamma(\alpha_D)\Gamma(\beta_D)} 
\times\int_0^\infty  \ln(x) \cdot x^{\frac{\alpha_D+\beta_D}{2}-1} K_{\alpha_D-\beta_D} \left(2 \sqrt{\alpha_D \beta_D x}\right) dx \notag\\
&= \psi(\alpha_D)+\psi(\beta_D) -\ln (\alpha_D\beta_D),
\end{align}
Combining \eqref{TPLR_turbulence_int} and \eqref{Ln_ID} yields Corollary \ref{turbulence_coro}.
\end{proof}

\begin{remark}
Considering $\psi(x)-\ln x \le 0$ for $x > 0$ in Corollary \ref{turbulence_coro}, we can conclude that the impact of atmospheric turbulence is always negative to TPLR in the $G_D$ saturation case. Further, the negative impact is quantified by the second part in \eqref{TPLR_Turbulence}, compared to the TPLR without atmospheric turbulence in Lemma \ref{TPLR_asym_lemma}.
\end{remark}

\subsection{Special Case for $\mu_V=\mu_H=0$ and $\sigma_V \neq \sigma_H$}
To get the exact closed-form expression for TPLR, we relax the conditions for the statistical characteristics of $\theta_V$ and $\theta_H$, i.e., $\mu_V=\mu_H=0$ and $\sigma_V \neq \sigma_H$. In this simplified case, $\theta^2$ follows the Hoyt distribution, and the corresponding TPLR is given by \cite{Juan_TIT}
\begin{align}
{\rm{TPLR}} = \frac{{2q}}{{1 + {q^2}}}Ie\left( {\frac{{1 - {q^2}}}{{1 + {q^2}}},\frac{{{{\left( {1 + {q^2}} \right)}^2}}}{{4{q^2}\left( {\sigma _V^2 + \sigma _H^2} \right)}}{\Theta _D}} \right),
\end{align}
where $Ie(\cdot,\cdot)$ denotes the Rice $Ie$-function defined  in \cite[Eq. (3)]{Juan_TIT}, and $q \in [0,1]$ is given by
\begin{align}
q = \begin{cases}
{{\sigma _V}/{\sigma _H},}&{\text{for }{\sigma _V} \le {\sigma _H};}\\
{{\sigma _H}/{\sigma _V},}&{\text{for }{\sigma _V} > {\sigma _H}.}
\end{cases}
\end{align}

From the asymptotic analysis for the general parameter settings, the asymptotic expression in the Hoyt distribution case ($\mu_V=\mu_H=0$) can be easily derived as
\begin{align}
{\rm{TPLR}}\mathop  \simeq \limits^{{\Theta _D} \to 0} \frac{{\sqrt \pi  }}{{4{\sigma _V}{\sigma _H}\Gamma \left( {1.5} \right)}}{\Theta _D}.
\end{align}
The asymptotic expression can be also easily derived by using the asymptotic result for the generalized Marcum Q-function in Proposition \ref{Q_asym_prop}, because the Rice $Ie$-function can be written in the Marcum Q-function form, given by \cite{Juan_TIT}
\begin{align}
&Ie\left( {k,x} \right) = \int_0^x {\exp \left( { - t} \right){I_0}\left( {kt} \right)dt}  \notag\\
&= \frac{1}{{\sqrt {1 - {k^2}} }}\left[ {{Q_1}\left( {\sqrt {\left( {1 + \sqrt {1 - {k^2}} } \right)x} ,\sqrt {\left( {1 - \sqrt {1 - {k^2}} } \right)x} } \right)} \right. \notag\\
&\hspace{5cm}\left. { - {Q_1}\left( {\sqrt {\left( {1 - \sqrt {1 - {k^2}} } \right)x} ,\sqrt {\left( {1 + \sqrt {1 - {k^2}} } \right)x} } \right)} \right].
\end{align}

When $q=1$, i.e., $\sigma_V=\sigma_H=\sigma$, the distribution of $\theta$ is reduced into the Rayleigh distribution, and the corresponding TPLR becomes
\begin{align}
{\rm{TPLR}} = 1 - \exp \left( { - \frac{{{\Theta _D}}}{{2{\sigma ^2}}}} \right).
\end{align}

\subsection{Special Case for $\mu_V \neq \mu_H \neq 0$ and $\sigma_V=\sigma_H=\sigma$}
For $\mu_V \neq \mu_H$ and $\sigma_V=\sigma_H=\sigma$, i.e., the Rice case, we can rewrite the TPLR as
\begin{align}
&{\rm{TPLR}} = \Pr \left\{ {\theta _V^2 + \theta _H^2 \le {\Theta _D}} \right\} \notag\\
&= \Pr \left\{ {{{\left( {\frac{{{\theta _V}}}{{\sqrt {\mu _V^2 + \mu _H^2} }}} \right)}^2} + {{\left( {\frac{{{\theta _H}}}{{\sqrt {\mu _V^2 + \mu _H^2} }}} \right)}^2} \le \frac{{{\Theta _D}}}{{\mu _V^2 + \mu _H^2}}} \right\} \notag\\
& = \Pr \left\{ {\sqrt {{{\left( {\frac{{{\theta _V}}}{{\sqrt {\mu _V^2 + \mu _H^2} }}} \right)}^2} + {{\left( {\frac{{{\theta _H}}}{{\sqrt {\mu _V^2 + \mu _H^2} }}} \right)}^2}}  \le \sqrt {\frac{{{\Theta _D}}}{{\mu _V^2 + \mu _H^2}}} } \right\}.
\end{align}
Let $\sin \beta  = \frac{{{\mu _V}}}{{\sqrt {\mu _V^2 + \mu _H^2} }}$ and $\cos \beta  = \frac{{{\mu _H}}}{{\sqrt {\mu _V^2 + \mu _H^2} }}$ It is easy to see that
\begin{align}
&\frac{{{\theta _V}}}{{\sqrt {\mu _V^2 + \mu _H^2} }} \sim \mathcal{N}\left( {\sin \beta ,\frac{{{\sigma ^2}}}{{\mu _V^2 + \mu _H^2}}} \right), \notag\\
&\frac{{{\theta _H}}}{{\sqrt {\mu _V^2 + \mu _H^2} }} \sim \mathcal{N}\left( {\cos \beta ,\frac{{{\sigma ^2}}}{{\mu _V^2 + \mu _H^2}}} \right). \notag
\end{align}
From the definition of the Rice distribution, we can conclude that
\begin{align}
\sqrt {{{\left( {\frac{{{\theta _V}}}{{\sqrt {\mu _V^2 + \mu _H^2} }}} \right)}^2} + {{\left( {\frac{{{\theta _H}}}{{\sqrt {\mu _V^2 + \mu _H^2} }}} \right)}^2}}
\sim {\rm Rice}\left( {1,\frac{\sigma }{{\sqrt {\mu _V^2 + \mu _H^2} }}} \right), \notag
\end{align}
and the corresponding TPLR can be easily derived by using the well-known standard cumulative distribution function (CDF) of the Rice distribution, given by
\begin{align}
{\rm TPLR} = 1 - {Q_1}\left( {\frac{{\sqrt {\mu _V^2 + \mu _H^2} }}{\sigma },\frac{{\sqrt {{\Theta _D}} }}{\sigma }} \right).
\end{align}
From the derivation in Proposition \ref{Q_asym_prop}, when $\Theta_D \to 0$, the asymptotic TPLR is
\begin{align}
{\rm TPLR} \mathop  \simeq \limits^{{\Theta _D} \to 0} \exp \left( { - \frac{{\mu _V^2 + \mu _H^2}}{{2{\sigma ^2}}}} \right)\frac{1}{{2{\sigma ^2}}}{\Theta _D} + o\left( {\Theta _D^2} \right).
\end{align}

\section{Transmission Probability Depending only  on Eavesdropper}\label{TPE_sec}
The transmission probability depending only on eavesdropper (TPE) is defined as the probability that the received power at the eavesdropper is less than a threshold ($\lambda_E$), i.e.,
\begin{align}\label{TPE_define}
{\rm{TPE}}&=\Pr \Big\{{K_2}{P_D}\left( \theta  \right)L\left( {{\theta _E}} \right){G_D} \le \lambda_E \Big\} 
= \Pr \Big\{ {{K_1}{K_2}G_D^2\exp \left( { - {G_D}\left( {2{\theta ^2} + 2\alpha {\theta _V} + {\alpha ^2}} \right)} \right) \le {\lambda _E}} \Big\} \notag\\
&= \Pr \Biggm\{ {{\theta ^2} + \alpha {\theta _V} \ge \underbrace {\frac{{ - 1}}{{2{G_D}}}\ln \frac{{{\lambda _E}}}{{{K_1}{K_2}G_D^2}} - \frac{{{\alpha ^2}}}{2}}_{{\Theta _E}}} \Biggm\}.
\end{align}
An approximation for the TPE will be given in Theorem \ref{TPE_theorem} based on the GHQ.

\begin{theorem}\label{TPE_theorem}
The TPE can be approximated as
\begin{align}
{\rm{TPE}} \approx \frac{1}{{\sqrt \pi  }}\sum\nolimits_{i = 1}^N {{\omega _i}{f_{{\rm{TPE}}}}\left( {{x_i}} \right)},
\end{align}
where $N$, $\omega_i$, and $x_i$ are the same as those in \eqref{TPLR_final}, and $f_{\rm TPE}(x)$ is given in \eqref{TPE_INT}, 
\begin{align}\label{TPE_INT}
f_{\rm TPE}(x) &={Q_{\frac{1}{2}}}\left( {\sqrt \lambda  ,\sqrt {\frac{{{\Theta _E} - {{\left( {\sqrt 2 {\sigma _V}x + {\mu _V}} \right)}^2} - \alpha \left( {\sqrt 2 {\sigma _V}x + {\mu _V}} \right)}}{{\sigma _H^2}}} } \right. \notag\\
&\hspace{4cm} \left. \times \mathbb{I}\left\{ {{\Theta _E} \ge {{\left( {\sqrt 2 {\sigma _V}x + {\mu _V}} \right)}^2} + \alpha \left( {\sqrt 2 {\sigma _V}x + {\mu _V}} \right)} \right\}\right).
\end{align}
\end{theorem}

\begin{proof}
See Appendix \ref{proof_TPE_theorem}.
\end{proof}

From \cite{Pan_TVT}--\cite{Zhao_CL}, for $\sigma _V^2 \ll \sigma _H^2$, the TPE can be robustly approximated by
\begin{align}
{\rm{TPE}}
\mathop  \approx  \limits^{\left( a \right)} & \frac{2}{3}\phi_{\rm TPE} \left( {{\mu _V}} \right) + \frac{1}{6}\phi_{\rm TPE} \left( {{\mu _V} + \sqrt 3 {\sigma _V}} \right) 
+ \frac{1}{6}\phi_{\rm TPE} \left( {{\mu _V} - \sqrt 3 {\sigma _V}} \right),
\end{align}
where $(a)$ follows \cite[Eq. (4)]{Pan_TVT},  and $\phi_{\rm TPE}(\cdot)$ is given by
\begin{align}
&\phi_{\rm TPE} \left( x \right)
= {Q_{\frac{1}{2}}}\left( {\sqrt \lambda  ,\sqrt {\frac{{{\Theta _E} - {x^2} - \alpha x}}{{\sigma _H^2}}\mathbb{I}\left\{ {{\Theta _E} \ge {x^2} + \alpha x} \right\}} } \right).
\end{align}

\subsection{Analysis on $\Theta_E$ versus $G_D$}
$\Theta_E$ is given by
\begin{align}
{\Theta _E} = \frac{{ - 1}}{{2{G_D}}}\ln \frac{{{\lambda _E}}}{{{K_1}{K_2}G_D^2}} - \frac{{{\alpha ^2}}}{2}.
\end{align}
The first derivative of $\Theta_E$ with respect to $G_D$ is
\begin{align}
\frac{{\partial {\Theta _E}}}{{\partial {G_D}}} = \frac{{\ln \frac{{{\lambda _E}}}{{{K_1}{K_2}G_D^2}} + 2}}{{2{G_D^2}}}.
\end{align}
Let $\frac{{\partial {\Theta _E}}}{{\partial {G_D}}}=0$. The root for $G_D>0$ is
\begin{align}
G_D^* = e\sqrt {\frac{{{\lambda _E}}}{{{K_1}{K_2}}}}.
\end{align}
We can conclude that $\Theta_E$ is an increasing function over $G_D \in (0,G_D^*)$, and a decreasing function over $G_D \in (G_D^*, \infty)$. The maximum value of $\Theta_E$ is
\begin{align}
\Theta _E^* = {\left. {{\Theta _E}} \right|_{{G_D} = G_D^*}} = \frac{{\sqrt {{K_1}{K_2}} }}{{e\sqrt {{\lambda _E}} }} - \frac{{{\alpha ^2}}}{2}.
\end{align}
For $G_D \to \infty$, we have
\begin{align}
\mathop {\lim }\limits_{{G_D} \to \infty } \frac{{ - 1}}{{2{G_D}}}\ln \frac{{{\lambda _E}}}{{{K_1}{K_2}G_D^2}} - \frac{{{\alpha ^2}}}{2} =  - \frac{{{\alpha ^2}}}{2}.
\end{align}
For $G_D \to 0^+$, we have the limit
\begin{align}
\mathop {\lim }\limits_{{G_D} \to {0^ + }} \frac{{ - 1}}{{2{G_D}}}\ln \frac{{{\lambda _E}}}{{{K_1}{K_2}G_D^2}} - \frac{{{\alpha ^2}}}{2} \to  - \infty.
\end{align}

\begin{remark}\label{trend_ThetaE}
Combining the analysis on $\Theta_E$ and \eqref{TPE_define}, we can conclude that the TPE is decreasing over $G_D \in (0,G_D^*)$ due to the increase in $\Theta_E$, while the TPE is increasing over $G_D \in (G_D^*,\infty)$ due to the decrease in $\Theta_E$. When $G_D$ is sufficiently large, the TPE will increase to an upper bound, because $\Theta_E$ goes to a constant ($-\alpha^2/2$). If $G_D \to 0^+$, $\Theta_E \to -\infty$, resulting in the TPE going to 1.
\end{remark}

\subsection{Asymptotic Analysis for TPE}
\begin{lemma}\label{TPE_asym_lemma}
When $\Theta_E \to \frac{-\alpha^2}{4}$ from the right side in the real number axis, the asymptotic result for TPE is given by
\begin{align}\label{TPE_Asym}
{\rm{TPE}} & \mathop  \simeq \limits^{{\Theta _E} \to \frac{{ - {\alpha ^2}}}{4}}  1 - \frac{{\sqrt \pi  \left( {{\alpha ^2} + 4{\Theta _E}} \right)}}{{16{\sigma _H}{\sigma _V}\Gamma \left( {1.5} \right)}}
\exp \left( { - \frac{\lambda }{2} - \frac{{{{\left( {\alpha /2 + {\mu _V}} \right)}^2}}}{{2\sigma _V^2}}} \right) 
\times \mathbb{I} \left\{ {{\Theta _E} \ge \frac{{ - {\alpha ^2}}}{4}} \right\}.
\end{align}
For $\Theta_E \le \frac{-\alpha^2}{4}$, the TPE is exactly equal to 1.
\end{lemma}

\begin{proof}
See Appendix \ref{proof_TPE_asym_lemma}.
\end{proof}

\begin{remark}
$\alpha$ reflects the distance between the sender and the eavesdropper. Specifically, $\alpha \to 0$ means that the eavesdropper is close to the sender. Lemma \ref{TPE_asym_lemma} presents a quantitative relationship between $\Theta_E$ (or $G_D$) and $\alpha$. When $\Theta_E \le \frac{-\alpha^2}{4}$, the event of received power at the eavesdropper below the pre-set threshold ($\lambda_E$) happens with probability one. Further, Lemma \ref{TPE_asym_lemma} reveals the linear trend of TPE with respect to $\Theta_E$  when $\Theta_E \to \left(\frac{-\alpha^2}{4}\right)^+$.
\end{remark}

\section{transmission probability depending on both legitimate receiver and eavesdropper}\label{TPRE_sec}
The transmission probability depending on both the legitimate receiver and the eavesdropper (TPRE) is defined as
\begin{align}\label{TPRE_Define}
&{\rm TPRE}=\Pr \Big\{ {{P_D}\left( \theta  \right) \ge {\lambda _D},{P_E}\left( {{\theta _E}} \right) \le {\lambda _E}} \Big\}  
= \Pr \Big\{ {{K_1}{G_D}L\left( \theta  \right) \ge {\lambda _D},{K_1}{K_2}G_D^2L\left( \theta  \right)L\left( {{\theta _E}} \right) \le {\lambda _E}} \Big\},
\end{align}
which can be further written by using the definition of $\Theta_D$ and $\Theta_E$,
\begin{align}\label{TPRE_Prob}
{\rm{TPRE}} &= \Pr \Big\{ {\theta _H^2 \le {\Theta _D} - \theta _V^2,\theta _H^2 \ge {\Theta _E} - \theta _V^2 - \alpha {\theta _V}} \Big\}  
= \Pr \left\{ {\frac{{{\Theta _E} - \theta _V^2 - \alpha {\theta _V}}}{{\sigma _H^2}} \le \frac{{\theta _H^2}}{{\sigma _H^2}} \le \frac{{{\Theta _D} - \theta _V^2}}{{\sigma _H^2}}} \right\}.
\end{align}
By applying the GHQ method, an approximate result for the TPRE is derived in Theorem \ref{TPRE_theorem}.

\begin{theorem}\label{TPRE_theorem}
A closed-form expression for an approximate TPRE is 
\begin{align}
{\rm{TPRE}} \approx \frac{1}{{\sqrt \pi  }}\sum\nolimits_{i = 1}^N {{\omega _i}{f_{{\rm{TPRE}}}}\left( {{x_i}} \right)},
\end{align}
where $N$, $\omega_i$, $x_i$ are the same as those in \eqref{TPLR_final}, and $f_{\rm TPRE} (\cdot)$ is given in \eqref{f_TPRE}, 
\begin{align}\label{f_TPRE}
&{f_{{\rm{TPRE}}}}\left( x \right) = \mathbb{I}\left\{ {x \ge \frac{{{\Theta _E} - {\Theta _D}}}{{\alpha \sqrt 2 {\sigma _V}}} - \frac{{{\mu _V}}}{{\sqrt 2 {\sigma _V}}}} \right\} 
\left[ {{Q_{\frac{1}{2}}}\Biggm( {\sqrt \lambda  , \frac{1}{\sigma _H}\sqrt { {{{\Theta _E} - {{\left( {\sqrt 2 {\sigma _V}x + {\mu _V}} \right)}^2} - \alpha \left( {\sqrt 2 {\sigma _V}x + {\mu _V}} \right)}}} }  } \right.  \notag \\
&\hspace{7cm}\times \mathbb{I}\left\{ {{\Theta _E} \ge {{\left( {\sqrt 2 {\sigma _V}x + {\mu _V}} \right)}^2} + \alpha \left( {\sqrt 2 {\sigma _V}x + {\mu _V}} \right)} \right\} \Biggm) \notag\\
&\hspace{2cm}\left. { - {Q_{\frac{1}{2}}}\Biggm( {\sqrt \lambda  , \frac{1}{\sigma _H}\sqrt {\left({\Theta _D} - {{\left( {\sqrt 2 {\sigma _V}x + {\mu _V}} \right)}^2}\right)\mathbb{I}\left\{ {{\Theta _D} \ge {{\left( {\sqrt 2 {\sigma _V}x + {\mu _V}} \right)}^2}} \right\}} } \Biggm)} \right].
\end{align}
\end{theorem}

\begin{proof}
See Appendix \ref{proof_TPRE_theorem}.
\end{proof}

For $\sigma_V^2 \ll \sigma_H^2$, according to \cite{Pan_TVT}--\cite{Zhao_CL}, the TPRE in \eqref{TPRE_exp} can be robustly approximated by
\begin{align}
{\rm{TPRE}}\mathop  \approx \limits^{\sigma _V^2 \ll \sigma _H^2} & \frac{2}{3}{\phi _{{\rm{TPRE}}}}\left( {{\mu _V}} \right) + \frac{1}{6}{\phi _{{\rm{TPRE}}}}\left( {{\mu _V} + \sqrt 3 {\sigma _V}} \right) 
+ \frac{1}{6}{\phi _{{\rm{TPRE}}}}\left( {{\mu _V} - \sqrt 3 {\sigma _V}} \right),
\end{align}
where ${\phi _{{\rm{TPRE}}}}\left( x \right)$ is given in  \eqref{phi_TPRE}, 
\begin{align}\label{phi_TPRE}
{\phi _{{\rm{TPRE}}}}\left( x \right) =& \mathbb{I}\left\{ {x \ge \frac{{{\Theta _E}}}{{2\alpha }} - \frac{{{\Theta _D}}}{\alpha }} \right\} 
 \left[ {{Q_{\frac{1}{2}}}\left( {\sqrt \lambda  ,\sqrt {\left( {\frac{{{\Theta _E}}}{{2\sigma _H^2}} - \frac{{{x^2} - \alpha x}}{{\sigma _H^2}}} \right)\mathbb{I}\left\{ {{\Theta _E} \ge 2\left( {{x^2} + \alpha x} \right)} \right\}} } \right)} \right. \notag\\
&\hspace{5cm}\left. { - {Q_{\frac{1}{2}}}\left( {\sqrt \lambda  ,\sqrt {\frac{{{\Theta _D} - {x^2}}}{{\sigma _H^2}}\mathbb{I}\left\{ {{\Theta _D} \ge {x^2}} \right\}} } \right)} \right]. 
\end{align}

\subsection{Asymptotic Analysis}
\begin{lemma}\label{TPRE_asym_lemma}
The asymptotic expression for TPRE as $G_D \to +\infty$ is the same as that for TPLR, given by
\begin{align}\label{TPRE_Asym}
{\rm TPRE}\mathop  \simeq \limits^{{G_D} \to  + \infty } \frac{{\sqrt \pi  {\Theta _D}}}{{4{\sigma _H}{\sigma _V}\Gamma \left( {1.5} \right)}}\exp \left( { - \frac{\lambda }{2} - \frac{{\mu _V^2}}{{2\sigma _V^2}}} \right).
\end{align}
\end{lemma}

\begin{proof}
The proof by solid mathematical manipulations is removed due to the page limitation.
\end{proof}

Lemma \ref{TPRE_asym_lemma} can be explained by the fact that for $G_D \to +\infty$, $\Theta_E \to \frac{-\alpha^2}{2}$, and $\Theta_D \to 0^+$. From the previous analysis, we know that for $\Theta_E < \frac{-\alpha^2}{4}$, there is no real root for $\Theta_E -\theta_V^2-\alpha \theta_V=0$, i.e., $\Theta_E-\theta_V^2-\alpha \theta_V$ is always less than zero for any $\theta_V$. From the second equal sign in   \eqref{TPRE_Prob}, for $\Theta_E < \frac{-\alpha^2}{4}$, the TPRE becomes
$
{\rm{TPRE}} = \Pr \left\{ {\frac{{\theta _H^2}}{{\sigma _H^2}} \le \frac{{{\Theta _D} - \theta _V^2}}{{\sigma _H^2}}} \right\},
$
which is exactly the definition of TPLR.

\begin{remark}
The asymptotic result for the TPRE in Lemma \ref{TPRE_asym_lemma}  reveals that although the increase in $G_D$ also induces an increase in the received power at the eavesdropper, the positive impact of increasing $G_D$ on the TPRE will domain the overall performance when $G_D \to \infty$, and more exactly, there is a convergence trend of TPLR and TPRE (i.e., the eavesdropping impact vanishes) in the $G_D$ saturation case.
\end{remark}

\subsection{Special Case for $\mu_V=\mu_H=0$ and $\sigma_V=\sigma_H=\sigma$}
As the TPRE cannot be solved in an exact closed-form, even for the simplest case, i.e., Rayleigh distribution. Here, we  only analyze the special case for the TPRE in the Rayleigh case ($\mu_V=\mu_H=0$ and $\sigma_V=\sigma_H=\sigma$).

Let $X = \theta _V^2 + \theta _H^2$ and $Y={2\theta _V^2 + 2\theta _H^2 + 2\alpha {\theta _V}}=2X+2 \alpha \theta_V$. The TPRE can be rewritten as
\begin{align}
{\rm{TPRE}} &= \Pr \Big\{ {{K_1}{G_D}\exp \left( { - {G_D}X} \right) \ge {\lambda _D},} \quad {{K_1}{K_2}G_D^2\exp \left( { - {G_D}\left( {Y + {\alpha ^2}} \right)} \right) \le {\lambda _E}} \Big\} \notag\\
& = \Pr \left\{ {X \le \frac{{ - 1}}{{{G_D}}}\ln \frac{{{\lambda _D}}}{{{K_1}{G_D}}}, \quad Y \ge \frac{{ - 1}}{{{G_D}}}\ln \frac{{{\lambda _E}\exp \left( {{G_D}{\alpha ^2}} \right)}}{{{K_1}{K_2}G_D^2}}} \right\}.
\end{align}
Let
\begin{align}
{\Theta _D} = \frac{{ - 1}}{{{G_D}}}\ln \frac{{{\lambda _D}}}{{{K_1}{G_D}}}, \notag \quad
{\Theta _E} = \frac{{ - 1}}{{{G_D}}}\ln \frac{{{\lambda _E}\exp \left( {{G_D}{\alpha ^2}} \right)}}{{{K_1}{K_2}G_D^2}}.\notag
\end{align}
The TPRE is further written as
\begin{align}
\Pr \left\{ {X \le {\Theta _D},Y \ge {\Theta _E}} \right\} = \int_0^{{\Theta _D}} {\int_{{\Theta _E}}^\infty  {{f_{X,Y}}\left( {x,y} \right)} } \dd y \dd x,
\end{align}
where $f_{X,Y}(\cdot,\cdot)$ represents the joint PDF of $X$ and $Y$. From the derivation of $f_{X,Y}(\cdot,\cdot)$ in Appendix \ref{PDF_appendix}, we can write the TPRE in an integral form as
\begin{align}\label{TPRE_Ray_int}
{\rm{TPRE}} =& \frac{1}{{4\pi {\sigma ^2}\alpha }}\int\nolimits_0^{{\Theta _D}} {\exp \left( { - \frac{x}{{2{\sigma ^2}}}} \right)}  
\int\nolimits_{\max \left\{ {{\Theta _E}, - 2\alpha \sqrt x  + 2x} \right\}}^{2x + 2\alpha \sqrt x } {{{\left[ {x - {{\left( {\frac{{y - 2x}}{{2\alpha }}} \right)}^2}} \right]}^{ - \frac{1}{2}}}} \dd y \dd x.
\end{align}
By considering the derivation in Appendix \ref{double_int_appendix}, the  TPRE can be easily approximated by
\begin{align}\label{TPRE_Ray}
{\rm{TPRE}} \approx \frac{{{\Theta _D}}}{{8\pi \alpha {\sigma ^2}}}\sum\nolimits_{i = 0}^N {{\omega _i^\prime}} {f^\prime_{{\rm{TPRE}}}}\left( {\frac{{{\Theta _D}}}{2}{x_i^\prime} + \frac{{{\Theta _D}}}{2}} \right),
\end{align}
where $\omega_i^\prime$ and $x_i^\prime$ are the weights and selected points over the standard integral interval (i.e., [-1,1]) in the Gaussian quadrature respectively, and ${f^\prime_{{\rm{TPRE}}}}\left( x \right)$ is given by
\begin{align}
&{f^\prime_{{\rm{TPRE}}}}\left( x \right) = \exp \left( { - \frac{x}{{2{\sigma ^2}}}} \right) 
\left[ {\alpha \pi  - 2\alpha \arcsin \left( {\frac{{\max \left\{ {{\Theta _E}, - 2\alpha \sqrt x  + 2x} \right\} - 2x}}{{2\alpha \sqrt x }}} \right)} \right].
\end{align}

\begin{remark}
There is no special function involved in \eqref{TPRE_Ray}, resulting in a much faster calculation for TPRE. Moreover, the built-in function for calculating the non-integer order generalized Marcum Q-function is not available in some softwares, such as Matlab, which brings more difficulty to implement Theorem \ref{TPRE_theorem}.
\end{remark}

\subsection{Further Simplification for the Rayleigh Case}
If we let ${\max \left\{ {{\Theta _E}, - 2\alpha \sqrt x  + 2x} \right\}}$ be always $- 2\alpha \sqrt x  + 2x$. The integral form of the TPRE in \eqref{TPRE_Ray_int} becomes
\begin{align}
{\rm{TPRE}} =& \frac{1}{{4\pi {\sigma ^2}\alpha }}\int_0^{{\Theta _D}} {\exp \left( { - \frac{x}{{2{\sigma ^2}}}} \right)}  \int_{ - 2\alpha \sqrt x  + 2x}^{2x + 2\alpha \sqrt x } {{{\left[ {x - {{\left( {\frac{{y - 2x}}{{2\alpha }}} \right)}^2}} \right]}^{ - \frac{1}{2}}}} \dd y \dd x.
\end{align}
By using the following integral identity,
\begin{align}
\int_{ - 2\alpha \sqrt x  + 2x}^{2x + 2\alpha \sqrt x } {{{\left[ {x - {{\left( {\frac{{y - 2x}}{{2\alpha }}} \right)}^2}} \right]}^{ - \frac{1}{2}}}} dy = 2\alpha \pi,
\end{align}
the closed-form expression for the TPRE can be finally derived as
\begin{align}\label{simple}
{\rm{TPRE}} = 1 - \exp \left( { - \frac{{{\Theta _D}}}{{2{\sigma ^2}}}} \right).
\end{align}

If $G_D \to +\infty$, we have
$
{\Theta _D} = \frac{1}{{{G_D}}}\ln \frac{{{K_1}{G_D}}}{{{\lambda _D}}} \to 0^+.
$
In this case ($\Theta_D \to 0^+$), by using $\exp \left( { - x} \right) \simeq 1 - x$ for $x \to 0$, the asymptotic result for TPRE is
\begin{align}
{\rm TPRE} \simeq \frac{{{\Theta _D}}}{{2{\sigma ^2}}} = \frac{1}{{2{\sigma ^2}{G_D}}}\ln \frac{{{K_1}{G_D}}}{{{\lambda _D}}}.
\end{align}
From the previous analysis, we know that the condition for using this very simple expression \eqref{simple} is $\Theta_E < \frac{-\alpha^2}{4}$. In this condition, the impact on the TPRE from the eavesdropper vanishes.

\section{Numerical Results}\label{numerical_sec}
\subsection{TPLR Simulations}
In this subsection, we run some Monte-Carlo simulations to validate the derived closed-form expressions for the TPLR. By referring to the parameter settings in \cite{Kupferman,Zhao_CL_QKD}, some selected simulation results are shown in Figs. \ref{TPLR_SigmaH}--\ref{TPLR_Robust}, where $10^7$ realizations are generated to get each average result according to the statistical  properties, i.e., $\mu_V$, $\mu_H$, $\sigma_V$ and $\sigma_H$.

In Fig. \ref{TPLR_SigmaH}, the growing trend of TPLR with decreasing $\sigma_H^2$ is obvious due to more weaker vibrations. As expected, the TPLR is decreasing as $\Theta_D$ approaches to zero (or equivalently $G_D \to +\infty$).
\begin{figure}[!htb]
\setlength{\abovecaptionskip}{0pt}
\setlength{\belowcaptionskip}{10pt}
\centering
\includegraphics[width= 3 in]{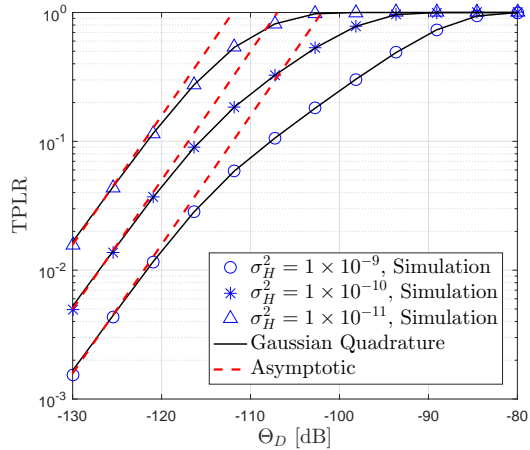}
\caption{TPLR versus $\Theta_D$ for $\mu_V=1 \times 10^{-8}$, $\mu_H=5 \times 10^{-8}$, $\sigma_V^2=1 \times 10^{-12}$ and $N=300$ for Gaussian quadrature.}
\label{TPLR_SigmaH}
\end{figure}

Figs. \ref{TPLR_Hoyt}--\ref{TPLR_Rice} plot the TPLR versus $\Theta_D$ for different variance cases in the Hoyt and Rice models respectively, where the trends with respect to $\Theta_D$ (or variance) are the same as those in Fig. \ref{TPLR_SigmaH}. The proposed asymptotic results  match the exact results very well especially in the low $\Theta_D$ region in Figs. \ref{TPLR_SigmaH}--\ref{TPLR_Rice}. It is worth noting that the asymptotic results are calculated much faster than the results by the Gaussian quadrature rule, which provides an alternative and efficient method, especially when we need a very high accuracy (several thousand terms may be needed in Gaussian quadrature method).

\begin{figure}[!htb]
\vspace{-0.5cm}
\setlength{\abovecaptionskip}{0pt}
\setlength{\belowcaptionskip}{10pt}
\centering
\includegraphics[width= 3 in]{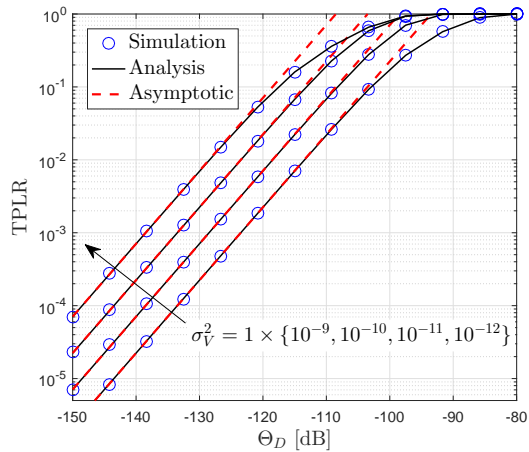}
\caption{TPLR in the Hoyt case versus $\Theta_D$ for $\mu_V=\mu_H=0$, and $\sigma_H^2= 5 \times 10^{-11}$.}
\label{TPLR_Hoyt}
\end{figure}

\begin{figure}[!htb]
\vspace{-0.5cm}
\setlength{\abovecaptionskip}{0pt}
\setlength{\belowcaptionskip}{10pt}
\centering
\includegraphics[width= 3 in]{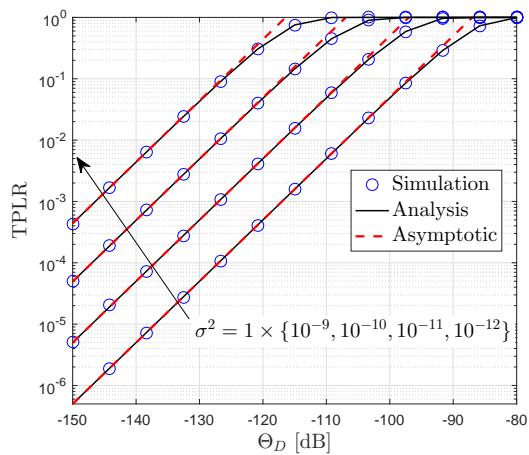}
\caption{TPLR in the Rice case versus $\Theta_D$ for $\sigma_V=\sigma_H=\sigma$, $\mu_V=1 \times 10^{-7}$ and $\mu_H=5 \times 10^{-7}$.}
\label{TPLR_Rice}
\end{figure}

The results shown in Fig. \ref{TPLR_Robust} validate the high accuracy of the robust approximation proposed in \cite{Pan_TVT}--\cite{Zhao_CL}, especially for $\sigma_V^2 \ll \sigma_H^2$. We can easily see that the gap between the exact and robust results almost vanishes as the ratio of $\sigma_V^2 / \sigma_H^2$ approaches to zero.

\begin{figure}[!htb]
\setlength{\abovecaptionskip}{0pt}
\setlength{\belowcaptionskip}{10pt}
\centering
\includegraphics[width= 3 in]{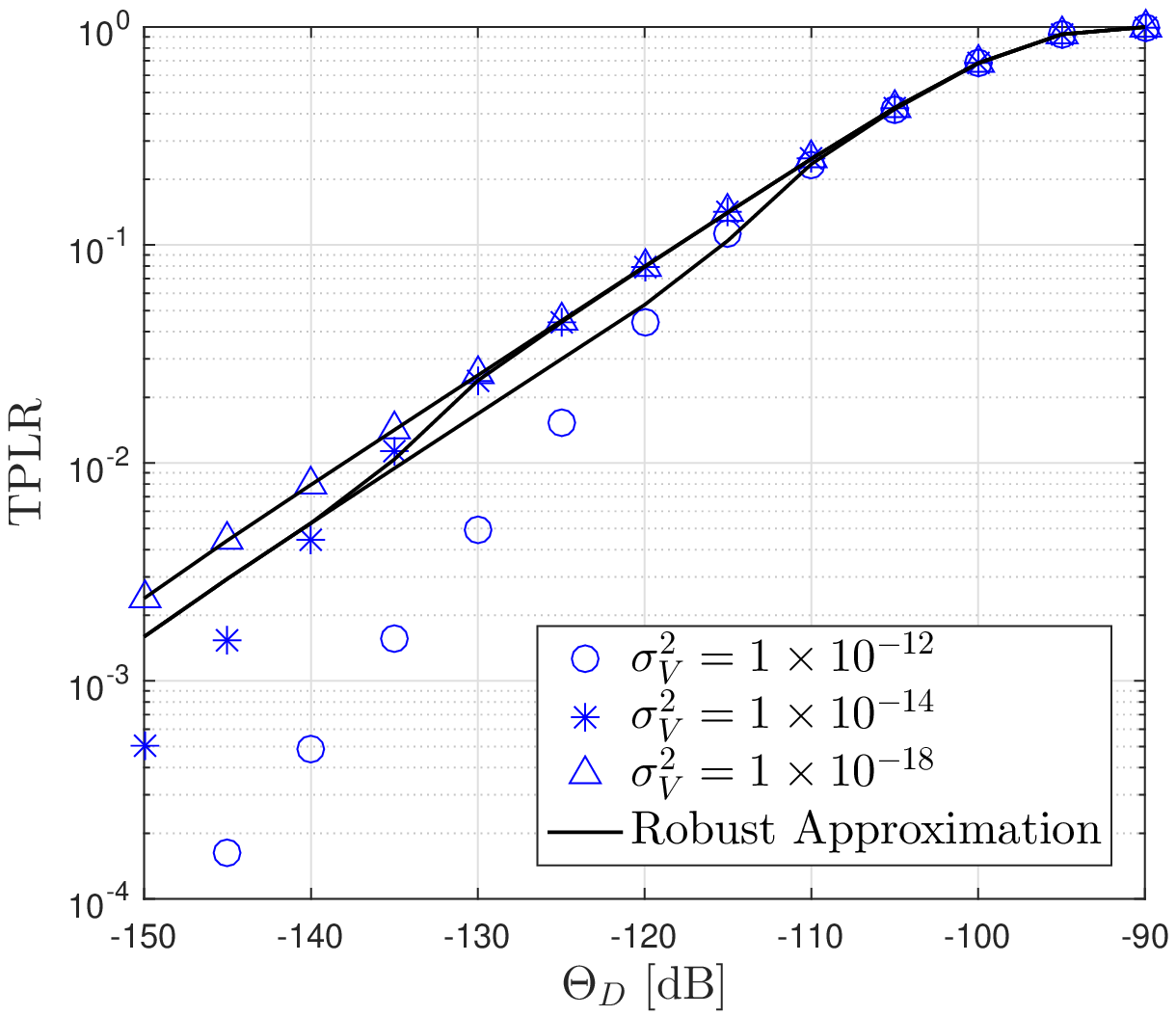}
\caption{TPLR versus $\Theta_D$ for $\sigma_H^2=1 \times 10^{-10}$, $\mu_V=1 \times 10^{-8}$ and $\mu_H=5 \times 10^{-8}$.}
\label{TPLR_Robust}
\end{figure}

\subsection{TPE Simulations}
In this subsection, we run the Monte-Carlo simulation to validate the correctness of the derived closed-form expressions for the TPE. To facilitate the simulation setting, we assume that $P_S=0$ dB, $G_S=G_E=10^{9}$, $\eta_S=\eta_D=\eta_E=0.9$, $\eta_q=0.1$, $\eta_B=0.04$, $\lambda_1=\lambda_2=780$ nm, $Z_1=Z_2=900$ km, and $L_A(D_1)=L_A(D_2)=0.5$, by considering Table I in \cite{Kupferman}. As the robust approximation for the Gaussian distribution has been investigated very well in \cite{Pan_TVT}--\cite{Zhao_CL}, we do not validate the high accuracy for $\sigma_V^2 \ll \sigma_H^2$ in the following simulations.

As shown in Fig. \ref{TPE_GD}, the TPE remains 1 before a sharp decrease up to the lowest bound, and after this lowest bound, the TPE grows rapidly to 1 as $G_D$ increases, which is exactly as the TPE changing analysis in Remark \ref{trend_ThetaE}. To show the lowest bound, Fig. \ref{TPE_GD_Log} uses the log-scale to plot TPE in Fig. \ref{TPE_GD}, where this lowest bound grows with increasing $\mu_H$. From Figs. \ref{TPE_GD}--\ref{TPE_GD_Log}, it is obvious that a large $\mu_H$ results in a large TPE, which can be explained by the fact that the received power performance at the legitimate receiver becomes worse, thereby accordingly decreasing the received power at the eavesdropper.

Fig. \ref{TPE_alpha} plots the TPE versus $\Theta_E$ from $-\frac{\alpha^2}{4}$ to $15 \alpha^2$. As analyzed in the IV-B subsection, for $\Theta_E \to -\frac{\alpha^2}{4}$ from the right side in the real number axis, the TPE can be approximated by a linear function. Moreover, there is a decreasing trend in the TPE with respect to $\alpha$.

\begin{figure}[!h]
\setlength{\abovecaptionskip}{0pt}
\setlength{\belowcaptionskip}{10pt}
\centering
\includegraphics[width= 3 in]{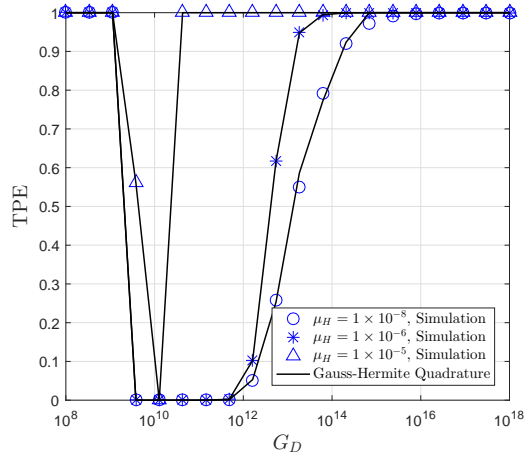}
\caption{TPE versus $G_D$ for $\mu_V=1 \times 10^{-7}$, $\sigma_V^2=1 \times 10^{-12}$, $\sigma_H^2=1 \times 10^{-13}$, $\lambda_E=1 \times 10^{-20}$, $\alpha=1 \times 10^{-9}$, and $N=30$ for the Gaussian quadrature.}
\label{TPE_GD}
\end{figure}

\begin{figure}[!h]
\setlength{\abovecaptionskip}{0pt}
\setlength{\belowcaptionskip}{10pt}
\centering
\includegraphics[width= 3 in]{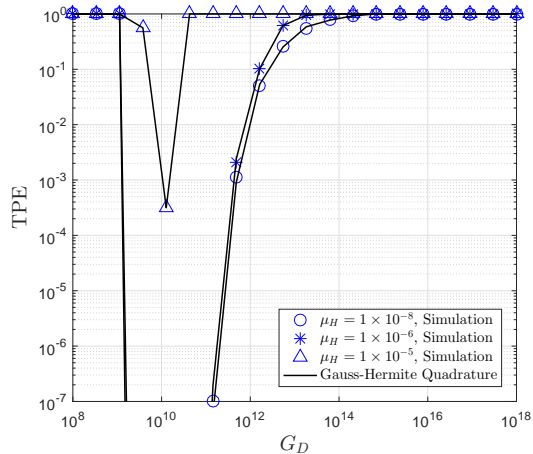}
\caption{TPE versus $G_D$ in the log-scale of Fig. \ref{TPE_GD}.}
\label{TPE_GD_Log}
\end{figure}

\begin{figure}
\setlength{\abovecaptionskip}{0pt}
\setlength{\belowcaptionskip}{10pt}
\centering
\includegraphics[width= 3 in]{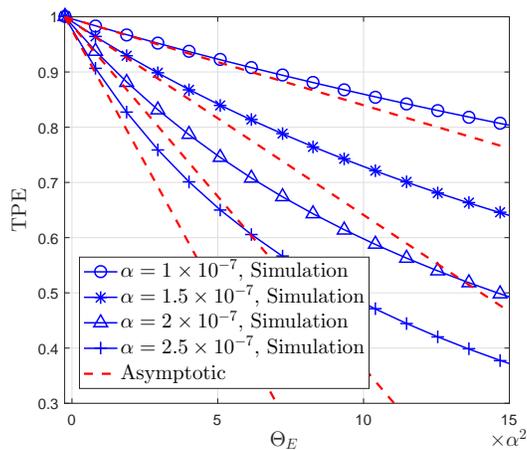}
\caption{TPE versus $\Theta_E$ for $\mu_V=1 \times 10^{-7}$, $\mu_H=1 \times 10^{-8}$, $\sigma_V^2 =1 \times 10^{-12}$, $\sigma_H^2=1 \times 10^{-13}$, and $\lambda_E=1 \times 10^{-20}$.}
\label{TPE_alpha}
\end{figure}

\subsection{TPRE Simulations}
In this subsection, the same system parameter settings as those in the first paragraph in the VI-B subsection are assumed for convenience.

In Fig. \ref{TPRE_alpha}, the TPRE versus $G_D$ is plotted, where a fluctuation is obvious in the medium $G_D$ region, while a monotonic decreasing is presented in the high $G_D$ region, and this decreasing trend can be approximated by a linear function (independent on $\alpha$, which has been proved in the asymptotic analysis for TPRE) in the log-scale. A smaller $\alpha$ means a more close eavesdropper around the sender, which results in the decrease in the TPRE.

\begin{figure}[!htb]
\setlength{\abovecaptionskip}{0pt}
\setlength{\belowcaptionskip}{10pt}
\centering
\includegraphics[width= 3 in]{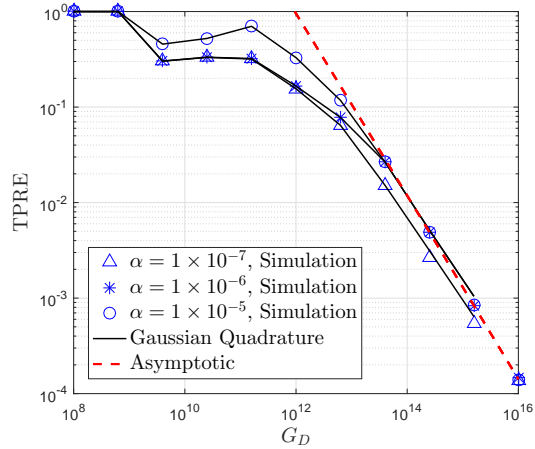}
\caption{TPRE versus $G_D$ for $\mu_V=1 \times 10^{-8}$, $\mu_H=5 \times 10^{-8}$, $\sigma_V^2 =1 \times 10^{-12}$, and $\sigma_H^2=1 \times 10^{-10}$, $\lambda_D=1 \times 10^{-15}$, $\lambda_E =1 \times 10^{-20}$, and $N=300$ for Gaussian quadrature.}
\label{TPRE_alpha}
\end{figure}

There also exists a fluctuation in $(\lambda_D, \lambda_E)=(10^{-15},10^{-20}), (10^{-20},10^{-20})$ cases in Fig. \ref{TPRE_Rayleigh}. In contrast, the TPRE for $(\lambda_D,\lambda_E)=(10^{-15},10^{-15})$ remains constant (equal to 1) before a monotonic decrease. From Fig. \ref{TPRE_Rayleigh}, we can also observe that the TPRE is an increasing function with respect to $\lambda_E$, and a decreasing function with respect to $\lambda_D$, which is easily explained by the joint probability properties in \eqref{TPRE_Define}.

\begin{figure}[!htb]
\setlength{\abovecaptionskip}{0pt}
\setlength{\belowcaptionskip}{10pt}
\centering
\includegraphics[width= 3 in]{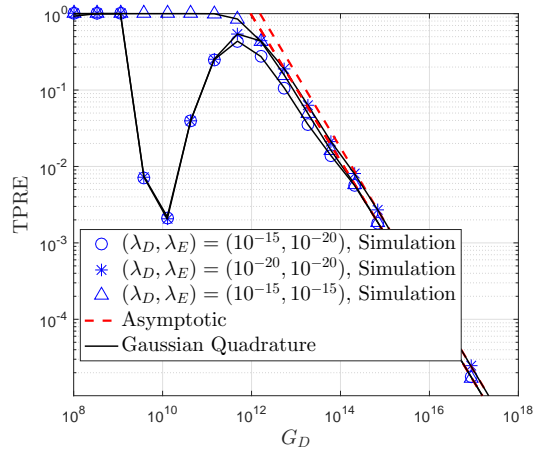}
\caption{TPRE versus $G_D$ for $\mu_V=\mu_H=0$,  $\sigma_V^2 =\sigma_H^2=1 \times 10^{-11}$ (Rayleigh distribution),  $\alpha=1 \times 10^{-6}$, and $N=110$.}
\label{TPRE_Rayleigh}
\end{figure}

\section{Conclusion}\label{conclude_sec}
In this paper, closed-form expressions for TPLR, TPE, and TPRE were derived based on the Gaussian quadrature rule, as well as the corresponding asymptotic formulas valid in the high telescope gain at the legitimate receiver region by using the asymptotic result for the generalized Marcum Q-function which could be also employed to derive the asymptotic expression for outage probability (showing the diversity order and array gain) over Beckmann fading channels. The approximate expressions bases on the Gaussian quadrature rule needs many terms for a high accuracy and the computation becomes more slower with increasing the terms. Alternatively, we can use the asymptotic expressions to get results with a very high accuracy in the high region of the telescope gain at the legitimate receiver.
Moreover, some closed-form expressions or more concise expressions in some special cases (like Rayleigh, Hoyt and Rice pointing error model cases) were also derived for much faster computations.

\appendices
\renewcommand{\thesectiondis}[2]{\Alph{section}:}
\section{proof of Theorem \ref{TPLR_theorem}}\label{proof_TPLR_theorem}
We can rewrite \eqref{TPLR_def} as
\begin{align}\label{TPLR_prob}
{\rm TPLR} &= \Pr \Big\{ {\theta _V^2 + \theta _H^2 \le {\Theta _D}} \Big\} = \Pr \Big\{ {\theta _H^2 \le {\Theta _D} - \theta _V^2} \Big\} 
= \Pr \left\{ {\frac{{\theta _H^2}}{{\sigma _H^2}} \le \frac{{{\Theta _D} - \theta _V^2}}{{\sigma _H^2}}} \right\},
\end{align}
where $\theta _H^2 / \sigma _H^2$ follows the non-central Chi-squared distribution with the unit degree, and the corresponding cumulative distribution function (CDF) is given by \cite{Non_Central}
\begin{align}
{F_{\theta _H^2/\sigma _H^2}}\left( x \right) = 1 - {Q_{\frac{1}{2}}}\left( {\sqrt \lambda  ,\sqrt {x\mathbb{I}\left\{ {x \ge 0} \right\}} } \right).
\end{align}

By using the CDF of $\theta _H^2 / \sigma _H^2$ and iterative expectation operation, the TPLR can be derived as
\begin{align}\label{TPLR_integral}
&{\rm{TPLR}} = {\mathbb{E}_{{\theta _V}}}\left\{ {1 - {Q_{\frac{1}{2}}}\left( {\sqrt \lambda  ,\sqrt {\frac{{{\Theta _D} - \theta _V^2}}{{\sigma _H^2}}\mathbb{I}\left\{ {{\Theta _D} \ge \theta _V^2} \right\}} } \right)} \right\} \notag\\
& \mathop  = \limits^{\left( a \right)} 1 - \frac{1}{{\sqrt {2\pi \sigma _V^2} }}\int_{ - \infty }^{ + \infty } {\exp \left( { - \frac{{{{\left( {{\theta _V} - {\mu _V}} \right)}^2}}}{{2\sigma _V^2}}} \right)} 
{Q_{\frac{1}{2}}}\left( {\sqrt \lambda  ,\sqrt {\frac{{{\Theta _D} - \theta _V^2}}{{\sigma _H^2}}\mathbb{I}\left\{ {{\Theta _D} \ge \theta _V^2} \right\}} } \right)\dd {\theta _V},
\end{align}
where $\mathbb{E}\{\cdot\}$ denotes the expectation operator, and $(a)$ follows the PDF of $\theta_V$, i.e.,  PDF of Gaussian distribution.
Let $x = \frac{{{\theta _V} - {\mu _V}}}{{\sqrt{2}\sigma _V}}$ in \eqref{TPLR_integral}, and then, we can get the integral form as \eqref{TPLR_Gaussian_Int},
\begin{align}\label{TPLR_Gaussian_Int}
{\rm{TPLR}} = 1 - \frac{1}{{\sqrt \pi  }} 
\int\limits_{ - \infty }^{ + \infty } {\exp \left( { - {x^2}} \right)}
{Q_{\frac{1}{2}}}\left( {\sqrt \lambda  ,\sqrt {\frac{{{\Theta _D} - {{\left( {\sqrt 2 {\sigma _V}x + {\mu _V}} \right)}^2}}}{{\sigma _H^2}}\mathbb{I}\left\{ {{\Theta _D} \ge {{\left( {\sqrt 2 {\sigma _V}x + {\mu _V}} \right)}^2}} \right\}} } \right) \dd x,
\end{align}
which can be approximated by using the Gaussian quadrature rule, shown in Theorem \ref{TPLR_theorem}. 

\section{Proof of Proposition \ref{Q_asym_prop}}\label{proof_marcumQ_prop}
The Marcum Q-function is defined  in the integral form,
\begin{align}
{Q_M}\left( {a,b} \right) = \int_b^\infty  {x{{\left( {\frac{x}{a}} \right)}^{M - 1}}\exp \left( { - \frac{{{x^2} + {a^2}}}{2}} \right)} {I_{M - 1}}\left( {ax} \right) \dd x,
\end{align}
where  $I_\cdot(\cdot)$ denotes the modified Bessel function of the first kind \cite{Gradshteyn}. To derive the asymptotic result for the Marcum Q-function, we first express $I_\cdot(\cdot)$ into the infinite series form,
\begin{align}\label{Bessel_fun}
{I_{M - 1}}\left( {ax} \right) = \sum\limits_{n = 0}^\infty  {\frac{1}{{n!\Gamma \left( {n + M} \right)}}} {\left( {\frac{{ax}}{2}} \right)^{2n + M - 1}}.
\end{align}
When $x \to 0$, it is easy to see that the leading term (i.e., the lowest order term of $x$) for approximating $I_\cdot(\cdot)$ in \eqref{Bessel_fun} is
\begin{align}
{\left. {{I_{M - 1}}\left( {ax} \right) \simeq \frac{1}{{n!\Gamma \left( {n + M} \right)}}{{\left( {\frac{{ax}}{2}} \right)}^{2n + M - 1}}} \right|_{n = 0}}
= \frac{1}{{\Gamma \left( M \right)}}{\left( {\frac{{ax}}{2}} \right)^{M - 1}}.
\end{align}
By using the asymptotic result for $I_\cdot(\cdot)$, the asymptotic result for the Marcum Q-function, as $b \to 0$, can be derived as
\begin{align}
&{Q_M}\left( {a,b} \right) = \underbrace {\int_0^\infty  {x{{\left( {\frac{x}{a}} \right)}^{M - 1}}\exp \left( { - \frac{{{x^2} + {a^2}}}{2}} \right)} {I_{M - 1}}\left( {ax} \right)\dd x}_{ = 1} \notag\\
&\hspace{2cm}-\int_0^b {x{{\left( {\frac{x}{a}} \right)}^{M - 1}}\exp \left( { - \frac{{{x^2} + {a^2}}}{2}} \right)} {I_{M - 1}}\left( {ax} \right)\dd x \notag\\
&\hspace{1cm}\mathop  \simeq \limits^{b \to 0} 1 - \frac{{\exp \left( { - \frac{{{a^2}}}{2}} \right)}}{{\Gamma \left( M \right){2^{M - 1}}}}\int_0^b {{x^{2M - 1}}\exp \left( { - \frac{{{x^2}}}{2}} \right)} \dd x,    
\end{align}
where the integral in the last step can be easily derived in a closed-form by using the definition of the lower incomplete Gamma function \cite{Gradshteyn}, given by
$
\int_0^b {{x^{2M - 1}}\exp \left( { - \frac{{{x^2}}}{2}} \right)} \dd x = {2^{M - 1}}\Upsilon \left( {M,\frac{{{b^2}}}{2}} \right),
$
where $\Upsilon(\cdot,\cdot)$ denotes the lower incomplete Gamma function. By using the  asymptotic result for $\Upsilon(\cdot,\cdot)$, i.e.,
$
\Upsilon \left( {s,x} \right) \to \frac{{{x^s}}}{s} \text{ as }x \to 0,
$
we can finally derive the asymptotic result for the Marcum Q-function as \eqref{Prop_31}.

\section{Proof of Lemma \ref{TPLR_asym_lemma}}\label{proof_TPLR_asym_lemma}
From the derivation in Proposition \ref{Q_asym_prop}, the CDF of the non-central Chi-squared distribution can be approximated by
\begin{align}
{F_{\theta _H^2/\sigma _H^2}}\left( x \right) &= 1 - {Q_{\frac{1}{2}}}\left( {\sqrt \lambda  ,\sqrt {x\mathbb{I}\left\{ {x \ge 0} \right\}} } \right) 
\mathop  \simeq \limits^{x \to 0} \frac{{\exp \left( { - \frac{\lambda }{2}} \right)}}{{\Gamma \left( {1.5} \right)\sqrt 2 }}{x^{0.5}}\mathbb{I}\left\{ {x \ge 0} \right\}.
\end{align}
Substituting the asymptotic expression for $F_{\theta_H^2 / \sigma_H^2}(\cdot)$ into \eqref{TPLR_prob} yeilds
\begin{align}\label{TPLR_asy1}
{\rm{TPLR}}  &\simeq  \frac{{\exp \left( { - \frac{\lambda }{2}} \right)}}{{\Gamma \left( {1.5} \right)\sqrt 2 }}{\mathbb{E}_{{\theta _V}}}\left\{ {{{\left( {\frac{{{\Theta _D} - \theta _V^2}}{{\sigma _H^2}}} \right)}^{\frac{1}{2}}}\mathbb{I}\left\{ {{\Theta _D} \ge \theta _V^2} \right\}} \right\} \notag\\
& = \frac{{\exp \left( { - \frac{\lambda }{2}} \right)}}{{\Gamma \left( {1.5} \right)\sqrt 2 }}\frac{1}{{\sqrt {2\pi \sigma _V^2} }} 
\int_{ - \infty }^{ + \infty } {\exp \left( { - \frac{{{{\left( {{\theta _V} - {\mu _V}} \right)}^2}}}{{2\sigma _V^2}}} \right)}
{\left( {\frac{{{\Theta _D} - \theta _V^2}}{{\sigma _H^2}}} \right)^{\frac{1}{2}}}\mathbb{I}\left\{ {{\Theta _D} \ge \theta _V^2} \right\}\dd {\theta _V} \notag\\
& = \frac{{\exp \left( { - \frac{\lambda }{2}} \right)}}{{\Gamma \left( {1.5} \right)\sqrt 2 }}\frac{1}{{\sqrt {2\pi \sigma _V^2} }} 
\int_{ - \sqrt {{\Theta _D}} }^{\sqrt {{\Theta _D}} } {\exp \left( { - \frac{{{{\left( {{\theta _V} - {\mu _V}} \right)}^2}}}{{2\sigma _V^2}}} \right)}
{\left( {\frac{{{\Theta _D} - \theta _V^2}}{{\sigma _H^2}}} \right)^{\frac{1}{2}}}\dd {\theta _V}.
\end{align}
As $\Theta_D \to 0$, the exponential function in the last step in \eqref{TPLR_asy1} can be approximated by
\begin{align}
\exp \left( { - \frac{{{{\left( {{\theta _V} - {\mu _V}} \right)}^2}}}{{2\sigma _V^2}}} \right)\mathop  \simeq \limits^{{\theta _V} \to 0} \exp \left( { - \frac{{{\mu _V^2}}}{{2\sigma _V^2}}} \right) + o\left( {\theta _V} \right),
\end{align}
Therefore, the asymptotic result for TPLR can be further written as
\begin{align}
{\rm{TPLR}} \simeq & \frac{{\exp \left( { - \frac{\lambda }{2}} \right)}}{{\Gamma \left( {1.5} \right)\sqrt 2 }}\exp \left( { - \frac{{{\mu _V^2}}}{{2\sigma _V^2}}} \right)\frac{\int\nolimits_{ - \sqrt {{\Theta _D}} }^{\sqrt {{\Theta _D}} } {{{\left( {\frac{{{\Theta _D} - \theta _V^2}}{{\sigma _H^2}}} \right)}^{\frac{1}{2}}}\dd {\theta _V}}}{{\sqrt {2\pi \sigma _V^2} }},
\end{align}
where the integral can be easily solved in a closed-form, 
\begin{align}
\int_{ - \sqrt {{\Theta _D}} }^{\sqrt {{\Theta _D}} } {{{\left( {\frac{{{\Theta _D} - \theta _V^2}}{{\sigma _H^2}}} \right)}^{\frac{1}{2}}}d{\theta _V}}  = \frac{\pi }{{2{\sigma _H}}}{\Theta _D}.
\end{align}
Finally, the asymptotic expression for TPLR is derived as \eqref{TPLR_asym}.

\section{Proof of Theorem \ref{TPE_theorem}}\label{proof_TPE_theorem}
The TPE can be  written as
\begin{align}
{\rm TPE} &= \Pr \Big\{ {\theta _V^2 + \alpha {\theta _V} + \theta _H^2 \ge {\Theta _E}} \Big\} 
= \Pr \left\{ {\frac{{\theta _H^2}}{{\sigma _H^2}} \ge \frac{{{\Theta _E} - \theta _V^2 - \alpha {\theta _V}}}{{\sigma _H^2}}} \right\},
\end{align}
where $\theta _H^2 / \sigma _H^2$ follows the non-central Chi-squared distribution with the unit degree and non-centrality parameter $\lambda$.
By using the CDF of $\theta_H^2 / \sigma_H^2$, we can get
\begin{align}\label{TPE_int1}
&{\rm{TPE}}=
{\mathbb{E}_{{\theta _V}}}\left\{ {{Q_{\frac{1}{2}}}\left( {\sqrt \lambda  ,\sqrt {\frac{{{\Theta _E} - \theta _V^2 - \alpha {\theta _V}}}{{\sigma _H^2}}\mathbb{I}\left\{ {{\Theta _E} \ge \theta _V^2 + \alpha {\theta _V}} \right\}} } \right)} \right\} \notag\\
& = \frac{1}{{\sqrt {2\pi \sigma _V^2} }}\int_{ - \infty }^{ + \infty } {\exp \left( { - \frac{{{{\left( {{\theta _V} - {\mu _V}} \right)}^2}}}{{2\sigma _V^2}}} \right)} {Q_{\frac{1}{2}}}\left( {\sqrt \lambda  ,\sqrt {\frac{{{\Theta _E} - \theta _V^2 - \alpha {\theta _V}}}{{\sigma _H^2}}\mathbb{I}\left\{ {{\Theta _E} \ge \theta _V^2 + \alpha {\theta _V}} \right\}} } \right)\dd {\theta _V}.
\end{align}
Let $x=\frac{\theta_V-\mu_V}{\sqrt{2} \sigma_V}$ in \eqref{TPE_int1}, and then we can get the following integral form,
\begin{align}
{\rm{TPE}} = \frac{1}{{\sqrt \pi  }}\int_{ - \infty }^{ + \infty } {\exp \left( { - {x^2}} \right)}  f_{\rm TPE} (x) \dd x,
\end{align}
which can be easily approximated by using the GHQ method, shown in Theorem \ref{TPE_theorem}.

\section{Proof of Lemma \ref{TPE_asym_lemma}}\label{proof_TPE_asym_lemma}
By using the asymptotic expression for $F_{\theta_H^2 / \sigma_H^2}(\cdot)$,
the TPE can be written as
\begin{align}\label{TPE_Asym1}
&{\rm{TPE}}
\simeq 1 - \frac{{\exp \left( { - \frac{\lambda }{2}} \right)}}{{\sqrt 2 \Gamma \left( {1.5} \right)}} 
{\mathbb{E}_{{\theta _V}}}\left\{ {{{\left( {\frac{{{\Theta _E} - \theta _V^2 - \alpha {\theta _V}}}{{\sigma _H^2}}} \right)}^{\frac{1}{2}}}\mathbb{I}\left\{ {{\Theta _E} \ge \theta _V^2 + \alpha {\theta _V}} \right\}} \right\},
\end{align}
where
\begin{align}\label{E_thetaV}
&{\mathbb{E}_{{\theta _V}}}\left\{ {{{\left( {\frac{{{\Theta _E} - \theta _V^2 - \alpha {\theta _V}}}{{\sigma _H^2}}} \right)}^{\frac{1}{2}}}\mathbb{I}\left\{ {{\Theta _E} \ge \theta _V^2 + \alpha {\theta _V}} \right\}} \right\} \notag\\
& = \frac{1}{{\sqrt {2\pi \sigma _V^2} }}\int_{ - \infty }^{ + \infty } {\exp \left( { - \frac{{{{\left( {{\theta _V} - {\mu _V}} \right)}^2}}}{{2\sigma _V^2}}} \right)}
{\left( {\frac{{{\Theta _E} - \theta _V^2 - \alpha {\theta _V}}}{{\sigma _H^2}}} \right)^{\frac{1}{2}}} 
\mathbb{I}\left\{ {{\Theta _E} \ge \theta _V^2 + \alpha {\theta _V}} \right\}\dd {\theta _V}.
\end{align}
To have a possible integral interval, the indicator function shows that we must have
\begin{align}
\theta _V^2 + \alpha {\theta _V} - {\Theta _E} \le 0.
\end{align}
To have real solutions for the above inequality, we  have
\begin{align}
{\alpha ^2} + 4{\Theta _E} \ge 0 \Longrightarrow {\Theta _E} \ge \frac{{ - {\alpha ^2}}}{4}.
\end{align}
If $\Theta_E < -\frac{-\alpha^2}{4}$, the indicator function in \eqref{E_thetaV} is always equal to zero. Therefore, the TPE is always 1, which is straightforward. In the following, we  consider the case of $\Theta_E \ge -\frac{-\alpha^2}{4}$.
Let
\begin{align}
{\theta _1} = \frac{{ - \alpha  + \sqrt {{\alpha ^2} + 4{\Theta _E}} }}{2}, \notag \quad
{\theta _2} = \frac{{ - \alpha  - \sqrt {{\alpha ^2} + 4{\Theta _E}} }}{2}. \notag
\end{align}
The integral in \eqref{E_thetaV} can be derived as \eqref{TPE_Asym2}, 
\begin{align}\label{TPE_Asym2}
&\mathbb{I}\left\{ {{\Theta _E} \ge \frac{{ - {\alpha ^2}}}{4}} \right\}\int\limits_{{\theta _2}}^{{\theta _1}} {\exp \left( { - \frac{{{{\left( {{\theta _V} - {\mu _V}} \right)}^2}}}{{2\sigma _V^2}}} \right)}
{\left( {\frac{{{\Theta _E} - \theta _V^2 - \alpha {\theta _V}}}{{\sigma _H^2}}} \right)^{\frac{1}{2}}}\dd {\theta _V}  \notag\\
&\mathop  \simeq \limits^{{\Theta _E} \to \frac{{ - {\alpha ^2}}}{4}} \mathbb{I}\left\{ {{\Theta _E} \ge \frac{{ - {\alpha ^2}}}{4}} \right\}\exp \left( { - \frac{{{{\left( {\alpha /2 + {\mu _V}} \right)}^2}}}{{2\sigma _V^2}}} \right)
\int\limits_{{\theta _2}}^{{\theta _1}} {{{\left( {\frac{{{\Theta _E} - \theta _V^2 - \alpha {\theta _V}}}{{\sigma _H^2}}} \right)}^{\frac{1}{2}}}\dd {\theta _V}} \notag\\
&= \mathbb{I}\left\{ {{\Theta _E} \ge \frac{{ - {\alpha ^2}}}{4}} \right\}\exp \left( { - \frac{{{{\left( {\alpha /2 + {\mu _V}} \right)}^2}}}{{2\sigma _V^2}}} \right)\frac{{\pi \left( {{\alpha ^2} + 4{\Theta _E}} \right)}}{{8{\sigma _H}}}.
\end{align}
Finally, combining \eqref{TPE_Asym1} and \eqref{TPE_Asym2} yields \eqref{TPE_Asym}.

\section{Proof of Theorem \ref{TPRE_theorem}}\label{proof_TPRE_theorem}
Using the iterative expectation method, i.e., $\mathbb{E}\left\{ {XY} \right\} = {\mathbb{E}_Y}\left\{ {{\mathbb{E}_{X\left| Y \right.}}\left( {X\left| Y \right.} \right)} \right\}$ in \eqref{TPRE_Prob}, we have
\begin{align}\label{TPRE_exp}
&{\rm{TPRE}} = {\mathbb{E}_{{\theta _V}}}\left\{ {\mathbb{I}\left\{ {\frac{{{\Theta _E} - \theta _V^2 - \alpha {\theta _V}}}{{\sigma _H^2}} \le \frac{{{\Theta _D} - \theta _V^2}}{{\sigma _H^2}}} \right\}} \right. \notag\\ 
&\hspace{3cm}\times\left[ {{F_{\theta _H^2/\sigma _H^2}}\left( {\frac{{{\Theta _D} - \theta _V^2}}{{\sigma _H^2}}} \right)} \right.\left. {\left. { - {F_{\theta _H^2/\sigma _H^2}}\left( {\frac{{{\Theta _E} - \theta _V^2 - \alpha {\theta _V}}}{{\sigma _H^2}}} \right)} \right]} \right\},
\end{align}
where
\begin{align}
&{F_{\theta _H^2/\sigma _H^2}}\left( {\frac{{{\Theta _D} - \theta _V^2}}{{\sigma _H^2}}} \right) - {F_{\theta _H^2/\sigma _H^2}}\left( {\frac{{{\Theta _E} - \theta _V^2 - \alpha {\theta _V}}}{{\sigma _H^2}}} \right) \notag\\
&= {Q_{\frac{1}{2}}}\left( {\sqrt \lambda  ,\sqrt {\frac{{{\Theta _E} - \theta _V^2 - \alpha {\theta _V}}}{{\sigma _H^2}}\mathbb{I}\left\{ {{\Theta _E} \ge \theta _V^2 + \alpha {\theta _V}} \right\}} } \right) 
- {Q_{\frac{1}{2}}}\left( {\sqrt \lambda  ,\sqrt {\frac{{{\Theta _D} - \theta _V^2}}{{\sigma _H^2}}\mathbb{I}\left\{ {{\Theta _D} \ge \theta _V^2} \right\}} } \right),
\end{align}
and
$\mathbb{I}\left\{ {\frac{{{\Theta _E} - \theta _V^2 - \alpha {\theta _V}}}{{\sigma _H^2}} \le \frac{{{\Theta _D} - \theta _V^2}}{{\sigma _H^2}}} \right\} = \mathbb{I}\left\{ {{\theta _V} \ge \frac{{{\Theta _E} - {\Theta _D}}}{\alpha }} \right\}.$
The TPRE is further written as \eqref{TPRE_int},
\begin{align}\label{TPRE_int}
{\rm{TPRE}} =& \frac{1}{{\sqrt {2\pi \sigma _V^2} }}\int_{ - \infty }^{ + \infty } {\left[ {{Q_{\frac{1}{2}}}\left( {\sqrt \lambda  ,\sqrt {\left( {\frac{{{\Theta _E} - \theta _V^2 - \alpha {\theta _V}}}{{\sigma _H^2}}} \right)\mathbb{I}\left\{ {{\Theta _E} \ge \theta _V^2 + \alpha {\theta _V}} \right\}} } \right)} \right.}  \notag\\
&\left. { - {Q_{\frac{1}{2}}}\left( {\sqrt \lambda  ,\sqrt {\frac{{{\Theta _D} - \theta _V^2}}{{\sigma _H^2}}\mathbb{I}\left\{ {{\Theta _D} \ge \theta _V^2} \right\}} } \right)} \right]
\exp \left( { - \frac{{{{\left( {{\theta _V} - {\mu _V}} \right)}^2}}}{{2\sigma _V^2}}} \right)\mathbb{I}\left\{ {{\theta _V} \ge \frac{{{\Theta _E} - {\Theta _D}}}{\alpha }} \right\}\dd {\theta _V}.
\end{align}
Let $x=\frac{\theta_V-\mu_V}{\sqrt{2}\sigma_V}$, and then, the integral form becomes \eqref{TPRE_INT}, 
\begin{align}\label{TPRE_INT}
{\rm{TPRE}} =& \frac{1}{{\sqrt \pi  }}\int_{ - \infty }^{ + \infty } {\exp \left( { - {x^2}} \right)\mathbb{I}\left\{ {x \ge \frac{{{\Theta _E} - {\Theta _D}}}{{\alpha \sqrt 2 {\sigma _V}}} - \frac{{{\mu _V}}}{{\sqrt 2 {\sigma _V}}}} \right\}} \notag\\
&\times \left[ {{Q_{\frac{1}{2}}}\left( {\sqrt \lambda  , \frac{1}{\sigma _H} \sqrt { {{\Theta _E} - {{\left( {\sqrt 2 {\sigma _V}x + {\mu _V}} \right)}^2} - \alpha \left( {\sqrt 2 {\sigma _V}x + {\mu _V}} \right)} } } \right.} \right. \notag\\
&\hspace{3cm}\times \mathbb{I}\left\{ {{\Theta _E} \ge {{\left( {\sqrt 2 {\sigma _V}x + {\mu _V}} \right)}^2} + \alpha \left( {\sqrt 2 {\sigma _V}x + {\mu _V}} \right)} \right\} \Bigg) \notag\\
&\left. { - {Q_{\frac{1}{2}}}\left( {\sqrt \lambda  , \frac{1}{\sigma _H} \sqrt {\left( {{\Theta _D} - {{\left( {\sqrt 2 {\sigma _V}x + {\mu _V}} \right)}^2}} \right) \mathbb{I}\left\{ {{\Theta _D} \ge {{\left( {\sqrt 2 {\sigma _V}x + {\mu _V}} \right)}^2}} \right\}} } \right)} \right]\dd x.
\end{align}
By using the GHQ method, an approximate result is shown in Theorem \ref{TPRE_theorem}.

\section{Derivation of the Joint PDF $f_{X,Y}(\cdot,\cdot)$}\label{PDF_appendix}
Let $Z=\theta_H^2$. We have $X = \theta _V^2 + Z$ and $Y = 2\theta _V^2 + 2Z + 2\alpha {\theta _V}$.
The PDF of $Z$ is
\begin{align}
{f_Z}\left( z \right) = \frac{1}{{\sqrt {2\pi {\sigma ^2}} }}{z^{ - \frac{1}{2}}}\exp \left( { - \frac{z}{{2{\sigma ^2}}}} \right), \quad z\ge 0.
\end{align}
The Jacobian matrix is
\begin{align}
&\left| {\begin{array}{*{20}{c}}
{\frac{{\partial X}}{{\partial {\theta _V}}}}&{\frac{{\partial X}}{{\partial Z}}}\\
{\frac{{\partial Y}}{{\partial {\theta _V}}}}&{\frac{{\partial Y}}{{\partial Z}}}
\end{array}} \right| = \left| {\begin{array}{*{20}{c}}
{2{\theta _V}}&1\\
{4{\theta _V} + 2\alpha }&2
\end{array}} \right|
= 4{\theta _V} - 4{\theta _V} - 2\alpha  =  - 2\alpha.
\end{align}
The joint PDF of $X$ and $Y$ is finally derived as
\begin{align}
&{f_{X,Y}}\left( {x,y} \right) = {f_{{\theta _V}}}\left( {{\theta _V}} \right){f_Z}\left( z \right){\left| {\frac{{\partial \left( {X,Y} \right)}}{{\partial \left( {{\theta _V},Z} \right)}}} \right|^{ - 1}} 
= \frac{1}{{2\alpha }}{f_{{\theta _V}}}\left( {\frac{{y - 2x}}{{2\alpha }}} \right){f_Z}\left( {x - {{\left( {\frac{{y - 2x}}{{2\alpha }}} \right)}^2}} \right) \notag\\
& = \frac{1}{{4\pi {\sigma ^2}\alpha }}\exp \left( { - \frac{x}{{2{\sigma ^2}}}} \right){\left[ {x - {{\left( {\frac{{y - 2x}}{{2\alpha }}} \right)}^2}} \right]^{ - \frac{1}{2}}} 
\mathbb{I}\Big\{ { - 2\alpha \sqrt x  + 2x \le y \le 2x + 2\alpha \sqrt x } \Big\}\mathbb{I}\Big\{ {x \ge 0} \Big\}.
\end{align}

\section{Derivation of the Double Integral in TPRE}\label{double_int_appendix}
In order to simplify the notations, we define the double integral in \eqref{TPRE_Ray_int} as $\mathcal{I}_1$.
By using the integral identity
\begin{align}
&\int\limits_{{\Theta _E}}^{2x + 2\alpha \sqrt x } {{{\left[ {x - {{\left( {\frac{{y - 2x}}{{2\alpha }}} \right)}^2}} \right]}^{ - \frac{1}{2}}}} \dd y
= \alpha \pi  - 2\alpha \arcsin \left( {\frac{{{\Theta _E} - 2x}}{{2\alpha \sqrt x }}} \right),
\end{align}
$\mathcal{I}_1$ can be written as
\begin{align}
{\mathcal{I}_1} =& \int_0^{{\Theta _D}} {\exp \left( { - \frac{x}{{2{\sigma ^2}}}} \right)}
\left[ {\alpha \pi  - 2\alpha \arcsin \left( {\frac{{\max \left\{ {{\Theta _E}, - 2\alpha \sqrt x  + 2x} \right\} - 2x}}{{2\alpha \sqrt x }}} \right)} \right] \dd x,
\end{align}
which can be easily calculated by using Gaussian quadrature method,
${\mathcal{I}_1} \approx \frac{{{\Theta _D}}}{2}\sum\limits_{i = 0}^N {{\omega _i^\prime}} {f^\prime_{{\rm{TPRE}}}}\left( {\frac{{{\Theta _D}}}{2}{x_i^\prime} + \frac{{{\Theta _D}}}{2}} \right).$



\ifCLASSOPTIONcaptionsoff
  \newpage
\fi


\begin{thebibliography}{1}

\bibitem{Shannon}
C. E. Shannon, ``Communication theory of secrecy systems," \emph{Bell Syst. Tech. J.}, vol. 28, no. 4, pp. 656--715, 1949.

\bibitem{Lo_Science}
H.-K. Lo, and H. F. Chau, ``Unconditional security of quantum key distribution over arbitrarily long distances," \emph{Science}, vol. 283, no. 5410, pp. 2050--2056, Mar. 1999.


\bibitem{Sasaki}
M. Sasaki, M. Fujiwara, R.-B. Jin, M. Takeoka, T. S. Han, H. Endo, K.-I. Yoshino, T. Ochi, S. Asami, and A. Tajima, ``Quantum photonic network: Concept, basic tools, and future issues," \emph{IEEE J. Sel. Topics  in Quantum Electron.}, vol. 21, no. 3, pp. 6400313, May, 2015.




\bibitem{Lo}
H.-K. Lo, M. Curty, and K. Tamaki, ``Secure quantum key distribution," \emph{Nat. Photonics}, vol. 8, pp. 595--604, 2014.

\bibitem{Tittel}
W. Tittel, ``Quantum key distribution breaking limits," \emph{Nat. Photonics}, vol. 13, pp. 310--311, 2019.



\bibitem{Inoue}
K. Inoue, ``Quantum key distribution technologies," \emph{IEEE J. Sel. Topics in Quantum Electron.}, vol. 12, no. 4, pp. 888--896, Aug. 2006.

\bibitem{Cheng}
C. Cheng, R. Lu, A. Petzoldt, and T. Takagi, ``Securing the internet of things in a quantum world," \emph{IEEE Commun. Mag.}, vol. 55, no. 2, pp. 116--120, Feb. 2017.

\bibitem{Selman}
A. Husagic-Selman, W. Al-Khateeb, and S. Saharudin, ``Feasibility of QKD over FSO link," in \emph{ Proc. International Conference on Computer and Communication Engineering (ICCCE 2012)},  July 2012, Kuala Lumpur, Malaysia, pp. 362--368.



\bibitem{Trinh}
P. V. Trinh, A. T. Pham, A. Carrasco-Casado, and M. Toyoshima, ``Quantum key distribution over FSO: Current development and
future perspectives," in \emph{Proc. Electromagnetics Research Symposium (PIERS-Toyama)}, Aug. 2018, pp. 1672--1679.

\bibitem{Imran}
I. S. Ansari, F. Yilmaz, and M.-S. Alouini, ``Performance analysis of free-space optical links over Malaga ($\mathcal{M}$) turbulence channels with pointing errors," \emph{IEEE Trans. Wireless Commun.}, vol. 15, no. 1, pp. 91--102, Jan. 2016.


\bibitem{R1}
H. Lei, H. Luo, K.-H. Park, I. S. Ansari,  W. Lei, G. Pan and M.-S. Alouini, ``On secure mixed RF-FSO systems with TAS and imperfect CSI,"  \emph{IEEE Trans. Commun.},	vol. 68, no. 7, pp. 4461--4475, Jul. 2020. 

\bibitem{R2}
K.-J. Jung, S. S. Nam, M.-S. Alouini, and Y.-C. Ko ``Unified statistical channel model of ship (or shore)-to-ship FSO communications with pointing errors," in \emph{Proc.  2019 IEEE Conference on Standards for Communications and Networking (CSCN)}, Oct. 2019, pp. 1--4.

\bibitem{R3}
 K.-J. Jung, S. S. Nam,  Y.-C. Ko, and M.-S. Alouini, ``BER Performance of FSO Links over Unified Channel Model for Pointing Error Models," in \emph{Proc. 2018 IEEE International Conference on Communications Workshops (ICC Workshops)}, May 2018, pp. 1--6.
 
 \bibitem{R4}
 H. Lei, H. Luo, K.-H. Park, Z. Ren, G. Pan, and M.-S. Alouini, ``Secrecy outage analysis of mixed RF-FSO systems with channel imperfection,"  \emph{IEEE Photon. J.}, vol. 10, no. 3, pp. 7904113, Jun. 2018.
 
 \bibitem{R5}
 H. Lei, Z. Dai, K.-H. Park, W. Lei, G. Pan, and M.-S. Alouini, ``Secrecy outage analysis of mixed RF-FSO downlink SWIPT systems,"  \emph{IEEE Trans. Commun.}, vol. 66, no. 12, pp. 6384--6395, Dec. 2018.
 
 \bibitem{R6}
 H. Lei, Z. Dai, I. S. Ansari, K.-H. Park, G. Pan, and M.-S. Alouini, ``On secrecy performance of mixed RF-FSO systems,"  \emph{IEEE Photon. J.}, vol. 9, no. 4, pp. 7904814, Aug. 2017.
 
 \bibitem{R7}
 I. S. Ansari, F. Yilmaz, and M.-S. Alouini, ``Impact of pointing errors on the performance of mixed RF/FSO dual-hop transmission systems,"  \emph{IEEE Wireless Commun. Lett.}, vol. 2, no. 3, pp. 351--354, May 2013.
 
 \bibitem{R8}
 M. A. Esmail, H. Fathallah, and M.-S. Alouini, ``Outage probability analysis of FSO links over Foggy channel,"  \emph{IEEE Photon. J.}, vol. 9, no. 2, pp. 7902312, Feb. 2017.
 
 



\bibitem{Xu}
F. Xu, B. Qi, and H.-K. Lo, ``Experimental demonstration of phase-remapping attack in a
practical quantum key distribution system," \emph{New J. Phys.}, vol. 12, pp. 113026, 2010.

\bibitem{Lydersen}
L. Lydersen, C. Wiechers, C. Wittmann, D. Elser, J. Skaar, and V. Makarov. ``Hacking commercial
quantum cryptography systems by tailored bright illumination," \emph{Nat. Photonics}, vol. 4, pp. 686--689, Aug. 2010.


\bibitem{Huang}
A. Huang, S. Sajeed, P. Chaiwongkhot, M. Soucarros, M. Legre, and V. Makarov,  ``Testing random-detector-efficiency countermeasure in a commercial system reveals a breakable
unrealistic assumption," \emph{IEEE J. Sel. in Quantum Electron.}, vol. 52, no. 11, pp. 8000211, Nov. 2016.


\bibitem{Meda}
A. Meda, I. P. Degiovanni, A. Tosi, Z. Yuan, G. Brida, and M. Genovese, ``Quantifying backflash radiation to prevent zero-error attacks in quantum key distribution," \emph{Light Sci. Appl.}, vol. 6, no. pp. e16261, 2017.

\bibitem{Weier}
H. Weier, H. Krauss, M. Rau, M. Furst, S. Nauerth, and H. Weinfurter, ``Quantum eavesdropping without
interception: an attack exploiting the dead time of single-photon detectors," \emph{New J. Phys.}, vol. 13, pp. 073024, Jul. 2011.

\bibitem{Shi}
Y. Shi, J. Z. J. Lim, H. S. Poh, P. K. Tan, P. A. Tan, A. Ling, and C. Kurtsiefer, ``Breakdown flash at telecom
wavelengths in InGaAs avalanche photodiodes," \emph{Opt. Express}, vol. 25, no. 24, pp. 30388--30394, Nov. 2017.

\bibitem{Lacaita}
A. Lacaita, S. Cova, A. Spinelli, and F. Zappa, ``Photon-assisted avalanche spreading in reach-through
photodiodes," \emph{Appl. Phys. Lett.}, vol. 62, no. 6, pp. 606--608, Feb. 1993.

\bibitem{Arnon}
S. Arnon, ``Effects of atmospheric turbulence and building sway on optical wireless-communication systems,"
\emph{Opt. Lett.}, vol. 28, no. 2, pp. 129--131, Jan. 2003.

\bibitem{Ansari}
I. S. Ansari, M.-S. Alouini, and J. Cheng, ``Ergodic capacity analysis of free-space optical links with nonzeros boresight pointing errors," \emph{IEEE Trans. Wireless Commun.}, vol. 14, no. 8, pp. 4248--4264, Aug. 2015.

\bibitem{Yang}
L. Yang, M.-S. Alouini, and I. S. Ansari, ``Asymptotic performance analysis of two-way relaying FSO networks with nonzero boresight pointing errors over double-generalized gamma fading channels," \emph{IEEE Trans Veh. Technol.}, vol. 67, no. 8, pp. 7800--7805, Aug. 2018.


\bibitem{Arnon_proc}
S. Arnon and N. S. Kopeika, ``Laser satellite communication network--vibration effect and possible solutions," \emph{Proc. IEEE}, vol. 85, no. 10, pp. 1646--1661, Oct. 1997.


\bibitem{Abdi}
A. Abdi, W. C. Lau, M.-S. Alouini,  and M. Kaveh, ``A new simple model for land mobile satellite channels: first- and second-order statistics," \emph{IEEE Trans. Wireless Commun.}, vol. 2, no. 3, pp. 519--528, May 2003.




\bibitem{Kupferman}
J. Kupferman, and S. Arnon, ``Zero-error attacks on a quantum key distribution FSO system," \emph{OSA Continuum}, vol. 1, no. 3, pp. 1079--1086, Nov. 2018.

\bibitem{Juan_TCOM}
J. P. Pena-Martin, J. M. Romero-Jerez, and F. J. Lopez-Martinez, ``Generalized MGF of Beckmann fading with applications to wireless communications performance analysis," \emph{IEEE Trans. Commun.}, vol. 65, no. 9, pp. 3933--3943, Sep. 2017.



\bibitem{Zhu}
B. Zhu, Z. Zeng, J. Cheng, and N. C. Beaulieu, ``On the distribution function of the generalized Beckmann random variable and its applications in communications," \emph{IEEE Trans. Commun.}, vol. 66, no. 5, pp. 2235--2250, May 2018.




\bibitem{Hessa_TWC}
H. AlQuwaiee, H.-C. Yang, and M.-S. Alouini, ``On the asymptotic capacity of dual-aperture FSO systems with generalzied pointing error model," \emph{IEEE Trans. Wireless Commun.}, vol. 15, no. 9, pp. 6502--6512, Sep. 2016.

\bibitem{Slim_book}
M. K. Simon and M.-S. Alouini, \emph{Digital Communication over Fading Channels}. New York, NY, USA: Wiley, 2000.

\bibitem{Zhao_CL_QKD}
H. Zhao and M.-S. Alouini, ``On the performance of quantum key distribution FSO systems under a generalized pointing error model," \emph{IEEE Commun. Lett.}, vol. 23, no. 10, pp. 1801--1805, Oct. 2019.


\bibitem{Nguyen}
H.-S. Nguyen, T.-S. Nguyen, and M. Voznak, ``Successful transmission probability of cognitive device-to-device communications underlaying cellular networks in the presence of hardware impairments," \emph{J. Wireless Com. Network}, no. 208 (2017), Dec. 2017. 

\bibitem{Gao}
J. Gao, ``On the successful transmission probability of cooperative cognitive radio ad hoc networks," \emph{Ad Hoc Netw.}, vol. 58, pp. 99--104, Apr. 2017.

\bibitem{Jiang}
Y. Jiang \emph{et al.}, ``Analysis and optimization of cache-enabled Fog radio access networks: Successful transmission probability, fractional offloaded traffic and delay," \emph{IEEE Trans. Veh. Technol.}, vol. 69, no. 5, pp. 5219--5231, May 2020.


\bibitem{Gaussian_Quadrature}
S. Venkateshan  and  P. Swaminathan, \emph{Computational  Methods  in  Engineering.}   Academic Press, 2014.

\bibitem{Non_Central}
D. J. Maširević, ``On new formulas for the cumulative distribution function of the noncentral Chi-square distribution," \emph{Mediterr. J. Math.}, vol. 14, no. 2, pp.  66 (2017),  Apr. 2017.

\bibitem{Gradshteyn}
I. S. Gradshteyn, I. M. Ryzhik, \emph{Table of Integrals, Series, and Products}, 7th edition. Academic Press, 2007.





\bibitem{Juan_TIT}
J. M. Romero-Jerez, and F. J. Lopez-Martinez, ``A new framework for the performance analysis of wireless communications under Hoyt (Nakagami-$q$) fading," \emph{IEEE Trans. Inf. Theory}, vol. 63, no. 3, pp. 1693--1702, Mar. 2017.





\bibitem{Pan_TVT}
G. Pan, C. Tang, X. Zhang, T. Li, Y. Weng, and Y. Chen, ``Physical-layer security over non-small-scale fading channels," \emph{IEEE Trans. Veh. Technol.}, vol. 65, no. 3, pp. 1326--1339,  Mar. 2016.

\bibitem{Holtzman}
J. M. Holtzman, ``A simple, accurate method to calculate spread multiple access error probabilities," \emph{IEEE Trans. Commun.}, vol. 40, no. 3, pp. 461--464, Mar. 1992.

\bibitem{Zhao_CL}
H. Zhao, Y. Liu, A. Sultan-Salem, and M.-S. Alouini,
``A simple evaluation for the secrecy outage probability over generalized-$K$ fading channels,"
\emph{IEEE Commun. Lett.}, vol. 23, no. 9, pp. 1479--1483, Sep. 2019.


\end{thebibliography}
\end{document}